\documentclass[12pt]{article}
\usepackage{color}
\usepackage{graphicx}
\usepackage{amsmath}
\usepackage{amssymb}
\usepackage{mathrsfs}
\usepackage[numbers,sort&compress]{natbib}
\usepackage{amsthm}
\numberwithin{equation}{section}
\usepackage{pgf}
\usepackage{CJK}
\usepackage{graphicx}
\usepackage{tikz}
\usepackage{stmaryrd}
\usepackage{graphicx}
\usepackage{hyperref}
\usepackage[latin1]{inputenc}
\usepackage[T1]{fontenc}
\usepackage[english]{babel}
\usepackage{listings}
\usepackage{xcolor,mathrsfs,url}
\usepackage{amssymb}
\usepackage{amsmath}
\usepackage{ifthen}
\newtheorem{theorem}{Theorem}
\newtheorem{lemma}{Lemma}
\newtheorem{corollary}{Corollary}
\newtheorem{Proposition}{Proposition}
\newtheorem{Assumption}{Assumption}
\usepackage {caption}
\usepackage{float}
\usepackage{subfigure}

\DeclareMathOperator*{\res}{Res}

\begin{document}

\title{ Long time asymptotic  behavior for  the  derivative   Schr$\ddot{o}$dinger equation with  nonzero boundary conditions }
\author{Yiling YANG$^1$, Qiaoyuan CHENG$^1$ and Engui FAN$^{1}$\thanks{\ Corresponding author and email address: faneg@fudan.edu.cn } }
\footnotetext[1]{ \  School of Mathematical Sciences, Fudan University, Shanghai 200433, P.R. China.}

\date{ }

\maketitle
\begin{abstract}
\baselineskip=17pt

In this paper, we  apply $\overline\partial$ steepest descent method  to study the Cauchy problem for the derivative nonlinear Schr$\ddot{o}$dinger equation with nonzero boundary conditions
\begin{align}
&iq_{t}+q_{xx}+i (|q|^2q)_{x}=0,  \nonumber\\
& q(x,0)=q_0(x), \nonumber
\end{align}	
 where  $\lim_{x\to\pm\infty}q_0(x)=q_\pm,\ \ |q_\pm|=1$.
Based on the spectral analysis of the Lax pair, we express    the  solution of the derivative  Schr$\ddot{o}$dinger equation
in terms of   solutions  of  a Riemann-Hilbert problem.
In a fixed space-time solitonic region $-3<x/t<-1$,
we   compute the long time asymptotic expansion of the solution $q(x,t)$,
which implies  soliton resolution conjecture and can
be characterized with  an $N(\Lambda)$-soliton whose parameters are modulated by
a sum of localized soliton-soliton
 interactions as one moves through the region; the  residual error order
 $\mathcal{O}( t^{-3/4})$ from a $\overline\partial$ equation.
{\bf AMS:} 35Q51; 35Q15; 37K15; 35C20.

\end{abstract}

\baselineskip=17pt

\newpage

\tableofcontents

\baselineskip=18pt
\section {Introduction}

\quad
The  study  on the long-time behavior of nonlinear wave equations which is solvable by the inverse scattering method was first  carried out by Manakov in 1974 \cite{Manakov1974}.
By using this method, Zakharov and Manakov   give the first result   for large-time asymptotic  of solutions for the  NLS equation  with  decaying initial data \cite{ZM1976}.
   The inverse scattering method    also    worked  for long-time behavior of integrable systems    such as  KdV,  Landau-Lifshitz  and the reduced Maxwell-Bloch   system \cite{SPC,BRF,Foka}.
     In 1993,  Deift and Zhou  developed a  nonlinear steepest descent method to rigorously obtain the long-time asymptotics behavior of the solution for the MKdV equation
by deforming contours to reduce the original Riemann-Hilbert (RH) problem   to a  model one  whose solution is calculated in terms of parabolic cylinder functions \cite{RN6}.
Since then    this method
has been widely  applied  to  the focusing NLS equation, KdV equation, Fokas-Lenells equation,  short-pulse equation and  Camassa-Holm equation  etc. \cite{RN9,RN10,Grunert2009,MonvelCH,xu2015,xusp}.

In recent years,   McLaughlin and   Miller further  presented a $\bar\partial$ steepest descent method which combine   steepest descent  with  $\bar{\partial}$-problem  rather than the asymptotic analysis
 of singular integrals on contours to analyze asymptotic of orthogonal polynomials with non-analytical weights  \cite{MandM2006,MandM2008}.
When  it  is applied  to integrable systems,   the $\bar\partial$ steepest descent method  also has  displayed some advantages,  such as   avoiding delicate estimates involving $L^p$ estimates  of Cauchy projection operators, and leading  the non-analyticity in the RH problem  reductions to a $\bar{\partial}$-problem in some sectors of the complex plane  which can be solved by being recast into an integral equation and by using Neumann series.   Dieng and  McLaughin used it to study the defocusing NLS equation  under essentially minimal regularity assumptions on finite mass initial data \cite{DandMNLS}; This $\bar\partial$ steepest descent method   was also successfully applied to prove asymptotic stability of N-soliton solutions to focusing NLS equation \cite{fNLS}; Jenkins et.al  studied  soliton resolution for the derivative nonlinear NLS equation  for generic initial data in a weighted Sobolev space \cite{Liu3}.  Their work provided  the   soliton resolution property  for  derivative NLS equation, which    decomposes  the solution into the sum of a finite number of separated solitons and a radiative parts when $t\to\infty$.  And  the dispersive part contains two components, one coming from the continuous spectrum and another from the interaction of the discrete and continuous spectrum.   For finite density initial data, Cussagna and  Jenkins studied the defocusing NLS equation \cite{SandRNLS}.

In this paper, we study the long time asymptotic behavior  for     the derivative nonlinear
 Schr$\ddot{o}$dinger (DNLS) equation with nonzero boundary conditions
\begin{align}
	&iq_{t}+q_{xx}+i\sigma(|q|^2q)_{x}=0,\label{DNLS} \\
	& q(x,0)=q_0(x),\hspace{0.5cm} \label{DNLS2}
\end{align}	
where  $\lim_{x\to\pm\infty}q_0(x)=q_\pm,\ \ |q_\pm|=1$.
Since the solution space of the  equation (\ref{DNLS})  with $ \sigma=1 $ and $ \sigma=-1 $ by the simple mapping $q(x,t)\to q(-x,t)$,
we only need to consider the case   $\sigma=-1$ in our paper.
 The DNLS equation as  a  completely integrable  system  was first proposed  by Kaup and Newell   \cite{KN1978}.

  The DNLS equation is often used to describe various nonlinear waves. For instance, DNLS equation governs  the evolution of small but finite amplitude
nonlinear Alfv$\acute{e}$n waves which propagates quasi-parallel to the magnetic field in space plasma physics \cite{R1971,Mj1976,Mi1976,Mj1989,Mj1997},
sub-picosecond pulses in single mode optical fibers \cite{DM1983,NM1981,GP}.  Moreover,  DNLS equation also describe  weak nonlinear electromagnetic waves in ferromagnetic \cite{N1991}, dielectric \cite{N1993}   and anti-ferromagnetic systems under external magnetic fields \cite{DV2002}.   Either zero boundary conditions   or nonzero boundary conditions   for the DNLS equation
 have well physically significant.  For problems of nonlinear Alfve waves, weak nonlinear electromagnetic waves in magnetic and dielectric media, waves propagating strictly parallel to the ambient magnetic fields are modeled by  zero boundary conditions,  while those oblique waves are modeled by the  nonzero boundary conditions.
 In optical fibers, pulses under bright background waves are modeled by the zero boundary conditions.
    Much work    on the DNLS equation    were also developed in \cite{ZH2007,KI1978,CL2004,C2006,L2007}.

     Zhang and  Yan presented the inverse scattering transform  of the DNLS equation (\ref{DNLS}) for both
     zero/nonzero boundary conditions  in terms of the matrix  RH  problems  \cite{ZGQ}.
For  Schwartz initial value $q_0\in\mathcal{S}(\mathbb{R})$,  Xu and Fan  derived the long-time asymptotic for (\ref{DNLS})  without soliton \cite{xf2013}
\begin{equation}
	q(x,t)=t^{-\frac{1}{2}}\alpha(\lambda_0)e^{\frac{ix^2}{4t}-i\nu(\lambda_0)\log t}+\mathcal{O}(t^{-1}\log t).
\end{equation}
   The long-time asymptotic for (\ref{DNLS})  with  step-like initial data
was further investigated \cite{Xf2014}.
Recently for generic initial data  in $H^{2,2} (\mathbb{R})$,   applying   $\bar\partial$ steepest descent method,   Jenkins et al  obtained
the following asymptotics for the  equation  (\ref{DNLS}) \cite{Liu3}
\begin{equation}
	q(x,t)=q_{sol}(x,t;D_I)+t^{-\frac{1}{2}}f(x,t)+\mathcal{O}(t^{-\frac{3}{4}}),
\end{equation}
where $q_{sol}(x,t;D_I)$ is the  soliton solutions of  the equation (\ref{DNLS}) with modulating reflectionless scattering
data. In our paper, for  finite density initial data $q_0-q_\pm \in H^{1,1}(\mathbb{R})$,    we  apply   $\bar\partial$ steepest descent method   to
    obtain  the  following    long-time asymptotic  of  the DNLS equation   (\ref{DNLS})
	\begin{align}
	q(x,t)=exp\left\lbrace \frac{i}{2}\int_{-\infty}^x( |q^r_\Lambda(x,t)|^2-1)dy \right\rbrace T(\infty)^{-2}q^r_\Lambda(x,t)+\mathcal{O}(t^{-3/4}).\label{resultu}
	\end{align}

This  paper is arranged as follows. In section \ref{sec2},
 we recall  some main  results on  the construction  process  of the  RH  problem with respect to  the initial problem of the DNLS equation  (\ref{DNLS})
  obtained in  \cite{CL2004,ZGQ},  which will be used
 to analyze   long-time asymptotics  of the DNLS equation in our paper.   In section \ref{sec3},   we   introduce
a function $T(z)$ to  define a new   RH problem  for  $M^{(1)}(z)$,  which  admits a regular discrete spectrum and  two  triangular  decompositions of the jump matrix.
   In section \ref{sec4},  by introducing a matrix-valued  function  $R(z)$,  we obtain  a mixed $\bar{\partial}$-RH problem  for  $M^{(2)}(z)$  by continuous extension of  $M^{(1)}(z)$.
     In section \ref{sec5},  we decompose  $M^{(2)}(z)$    into a
 model RH   problem  for  $M^{(r)}(z)$ and a  pure $\bar{\partial}$ Problem for  $M^{(3)}(z)$.
 The  $M^r(z)$  can be obtained  via  an modified reflectionless RH problem $M^{(r)}_\Lambda(z)$   for the soliton components which  is solved   in Section \ref{sec6}.
  In section \ref{sec7},   the error function  $E(z)$ between $M^{(r)}(z)$ and $M^{(r)}_\Lambda(z)$ can be computed  with a  small-norm RH problem.
 In Section \ref{sec8},   we analyze  the $\bar{\partial}$-problem  for $M^{(3)}(z)$.
   Finally, in Section \ref{sec9},   based on  the result obtained above,   a relation formula
   is found
\begin{align}
 M(z) = T(\infty)^{\sigma_3}M^{(3)}(z)E(z)M^{(r)}_\Lambda(z)R^{(2)}(z)^{-1}T(z)^{-\sigma_3},\nonumber
\end{align}
from which   we then obtain the   long-time   asymptotic behavior  for the DNLS equation (\ref{DNLS}) via a reconstruction formula.

\section {The spectral analysis and    a RH problem}\label{sec2}

\quad The DNLS equation (\ref{DNLS}) is completely integrable and    admits the Lax pair \cite{KN1978}
\begin{equation}
\Phi_x = X \Phi,\hspace{0.5cm}\Phi_t =T \Phi, \label{lax0}
\end{equation}
while
\begin{equation}
X=ik^2 \sigma_3+kQ,\nonumber
\end{equation}
\begin{equation}
T=-\left(2k^2+Q^2 \right)X-ikQ_x\sigma_3,\nonumber
\end{equation}
where $k\in \mathbb{C}$ is a spectral parameter and
\begin{equation}
	Q=\left(\begin{array}{cc}
		0 & q  \\
		-\bar{q} & 0
	\end{array}\right), \ \ \sigma_3=\left(\begin{array}{cc}
1 & 0   \\
0 & -1
\end{array}\right).
\end{equation}

By using  the boundary condition  (\ref{DNLS}),    the Lax pair (\ref{lax0})  becomes
\begin{equation}
	\Phi_{\pm,x}\sim X_\pm\Phi_\pm,\hspace{0.5cm}\Phi_{\pm,t}\sim T_\pm\Phi_\pm, \ \ \  \ x\rightarrow \pm \infty,\label{asymptoticlax}
\end{equation}
where
\begin{equation}
	X_\pm=ik^2 \sigma_3+kQ_\pm,\hspace{0.5cm}T_\pm=-\left(2k^2-1 \right)X_\pm,
\end{equation}
and
\begin{equation*}
	Q_\pm=\left(\begin{array}{cc}
		0 & q_\pm \\
		-\bar{q}_\pm & 0
	\end{array}\right).
\end{equation*}

The eigenvalues of the matrix $X_\pm$ are $\pm ik\lambda$, which  satisfy
\begin{equation}
	\lambda^2=k^2+1.
\end{equation}
To avoid multi-valued case  of   eigenvalue  $\lambda$,   we  introduce    a uniformization variable
\begin{equation}
	z= k+\lambda,
\end{equation}
and  obtain two single-valued functions
\begin{equation}
	k(z)=\frac{1}{2}(z-\frac{1}{z}),\hspace{0.5cm}\lambda(z)=\frac{1}{2}(z+\frac{1}{z}).\label{uniformization55}
\end{equation}
Define two domains $D^+$, $D^-$ and their boundary   $\Sigma$ on $z$-plane by
\begin{align*}
	&D^-=\{z:{\rm Re}z{\rm Im}z<0\},\hspace{0.5cm}D^+=\{z:{\rm Re}z{\rm Im}z>0\},\\
	&  \Sigma=\{z:{\rm Re}z{\rm Im}z=0\}=  \mathbb{R} \cup i\mathbb{R}\backslash \{0\},
\end{align*}
which   are shown in Figure \ref{fig:figure2}.
\begin{figure}[H]
	\centering
		\begin{tikzpicture}[node distance=2cm]
		\filldraw[yellow!40,line width=2] (2.4,0.01) rectangle (0.01,2.4);
		\filldraw[yellow!40,line width=2] (-2.4,-0.01) rectangle (-0.01,-2.4);
		\draw[->](-2.5,0)--(3,0)node[right]{$\mathbb{R}$};
		\draw[->](0,-2.5)--(0,3)node[above]{$i\mathbb{R}$};
		\draw[->](0,0)--(-0.8,0);
		\draw[->](-0.8,0)--(-1.8,0);
		\draw[->](0,0)--(0.8,0);
		\draw[->](0.8,0)--(1.8,0);
		\coordinate (A) at (0.5,1.2);
		\coordinate (B) at (0.6,-1.2);
		\coordinate (G) at (-0.6,1.2);
		\coordinate (H) at (-0.5,-1.2);
		\coordinate (I) at (0.16,0);
		\fill (A) circle (0pt) node[right] {$D^+$};
		\fill (B) circle (0pt) node[right] {$D_-$};
		\fill (G) circle (0pt) node[left] {$D_-$};
		\fill (H) circle (0pt) node[left] {$D^+$};
		\fill (I) circle (0pt) node[below] {$0$};
	\end{tikzpicture}
	\caption{ The domains  $D^-$, $D^+$  and boundary  $\Sigma=\mathbb{R}\cup i\mathbb{R}\backslash \{0\}$.}
	\label{fig:figure2}
\end{figure}
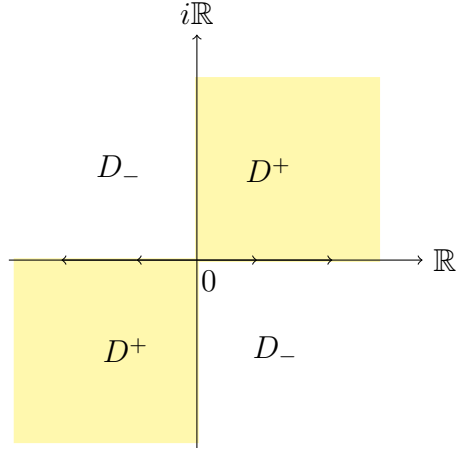

We  derive the solution of the asymptotic spectral problem (\ref{asymptoticlax})
\begin{equation}
	\Phi_\pm\sim
		Y_\pm e^{ik(z)\lambda(z)x\sigma_3},
\end{equation}
where
\begin{equation*}
	Y_\pm=\left(\begin{array}{cc}
		1 & \frac{iq_\pm}{z} \\
		\frac{i\bar{q}_\pm}{z} & 1
	\end{array}\right).
\end{equation*}
By making transformation
\begin{equation}
	\mu_\pm=\Phi_\pm e^{-ik\lambda x\sigma_3},\label{trans}
\end{equation}
then  we have
\begin{align}
&\mu_\pm \sim Y_\pm, \hspace{0.5cm} x \rightarrow \pm\infty,\nonumber\\
&\det[\Phi_\pm]=\det[\mu_\pm] =\det [Y_\pm]=1+z^{-2},\nonumber
\end{align}
and $\mu_\pm$  satisfy the Volterra integral equations
\begin{align}
	&\mu_\pm(z)=Y_\pm+\int_{\pm\infty}^{x}Y_\pm e^{ik\lambda(x-y){\hat{\sigma}}_3}[Y_\pm^{-1}\Delta X_\pm\mu_\pm(z)]dy, \ \ \ z\not= \pm i, \ \ \label{jost1}
\end{align}
\begin{align}
	&\mu_\pm(z)=Y_\pm+\int_{\pm\infty}^{x}\left[I+(x-y)X_\pm(z) \right] \Delta X_\pm\mu_\pm(z)dy,
	\ \ \ z=\pm i,\label{jost2}
\end{align}
where $\Delta X_\pm=k\left(Q-Q_\pm \right) $.

It can be shown that the eigenfunction $\mu_\pm$ admit  symmetry   \cite{ZGQ}.
\begin{Proposition}\label{sym}
	Jost functions  admit two  reduction conditions on
	the $z$-plane:
	
	The first symmetry reduction:
\begin{equation}
\mu_\pm(z)=\sigma_2\overline{\mu_\pm(\bar{z})}\sigma_2=\sigma_1\overline{\mu_\pm(-\bar{z})}\sigma_1.\label{symPhi1}
\end{equation}

The second symmetry reduction:
\begin{equation}
	\mu_\pm(z)=\frac{i}{z}\mu_\pm(-z^{-1})\sigma_3Q_\pm.\label{symPhi2}
\end{equation}
\end{Proposition}

And for $z\in \Sigma^0=\Sigma\setminus\left\lbrace \pm i \right\rbrace $, there exist scattering matrix which is a linear  relation between $\Phi_+$ and $\Phi_-$
\begin{equation}
	\Phi_+(x,t,z)=\Phi_-(x,t,z)S(z), \label{scattering}
\end{equation}
where
\begin{equation}
		S(z) =\left(\begin{array}{cc}
			a(z) &-\overline{b(\bar{z})}   \\[4pt]
			 b(z) & \overline{a(\bar{z})}
		\end{array}\right),\hspace{0.5cm}\det[S(z)]=1
\end{equation}
with symmetry reduction:
\begin{equation}
	S(z)=\sigma_1\overline{S(-\bar{z})}\sigma_1=(\sigma_3Q_-)^{-1}S\left( -z^{-1}\right) \sigma_3Q_+.\label{symS}
\end{equation}
And   the reflection coefficients are defined by
\begin{equation}
	\rho(z)=\frac{b(z)}{a(z)},\hspace{0.5cm}\tilde{\rho}(z)=-\overline{\rho(\bar{z})},\label{symrho}
\end{equation}
with symmetry reduction:
\begin{equation}
	\rho(z)=\overline{\tilde{\rho}(-\bar{z})}=\frac{\bar{q}_-}{q_-}\tilde{\rho}(-z^{-1}).
\end{equation}
Then
\begin{align}
	&a(z)=\frac{{\rm Wr}(\Phi_{+}^1,\Phi_{-}^2)}{1+z^{-2}},\hspace{0.5cm}b(z)=\frac{{\rm Wr}(\Phi_-^1,\Phi_{+}^1)}{1+z^{-2}}.\label{scatteringcoefficient2}
\end{align}
Although $a(z)$ and $b(z)$ has singularities at points $\pm i$, $|\rho(\pm i)|=1$.
The uniqueness and existences of Lax pair  from \cite{ZGQ}:
\begin{Proposition}\label{proasyM}
	If $ q-q_\pm\in L^{1,1}(\mathbb{R}_\pm)$, the fundamental eigenfunctions $\mu_\pm$
	defined by (\ref{jost1}) and (\ref{jost2})  exist and is the unique. Define $\mu_\pm=(\mu_{\pm}^1,\mu_{\pm}^2)$ with   $\mu_{\pm}^1$ and $ \mu_{\pm}^2$ denoting the first and second   column of $\mu_\pm$
	respectively. Then $\mu_+^1$and  $ \mu_-^2$ are analytical on the $D^+$ , and continuous in  $\overline{D^+}$; $\mu_-^1$ and $ \mu_+^2$ are analytical on the $D^-$, and continuous in    $\overline{D^-}$.  Moreover, form (\ref{scatteringcoefficient2}), $a(z)$ is analytical on the $D^+$ , and continuous in  $\overline{D^+}\setminus \left\lbrace \pm i \right\rbrace$. Further, $\lambda a(z)$ is analytical on the $D^+$ , and continuous in  $\overline{D^+}$. $b(z)$ and $\lambda b(z)$ are continuous in  $\Sigma^0$ and $\Sigma$ respectively.

\end{Proposition}

The zeros of $a(z)$ on $\Sigma$   are known to
 occur and they correspond to spectral singularities.  They are excluded from our analysis in the this paper. To deal with our following work,
we assume our initial data satisfy this assumption.
\begin{Assumption}\label{initialdata}
	The initial data $q-q_\pm (x) \in L^{1,1}(\mathbb{R}^\pm)$     and it generates generic scattering data which satisfy that
	
	\textbf{1. }a(z) has no zeros on $\Sigma$.
	
	\textbf{2. }a(z) only has finite number of simple zeros.
	
	\textbf{3. } $\rho(z)$ and $\tilde{\rho}(z)$ belong to $W^{2,\infty}(\Sigma)\cap W^{1,2}(\Sigma)$.
\end{Assumption}
Suppose that $a(z)$ has $N_1$ simple zeros $z_1,...,z_{N_1}$ on $D^+\cap\{z\in\mathbb{C}:{\rm Im}z>0,|z|>1\}$, and $N_2$ simple
zeros $w_1,...,w_m$ on the circle $\{z=e^{i\varphi}: 0<\varphi<\frac{\pi}{2} \}$.   The  symmetries  (\ref{symS}) imply that
\begin{align}
&a(\pm z_n)=0 \Longleftrightarrow \overline{a(\pm \bar{z}_n)}=0   \Longleftrightarrow   \overline{a\left(\pm \frac{1}{z_n}\right)}=0\nonumber\\
&  \Longleftrightarrow   a\left(\pm \frac{1}{\bar{z}}\right)=0, \hspace{0.5cm}n=1,...,N_1,\nonumber
\end{align}
and on the circle
\begin{equation}
	a(\pm w_m)=0\Longleftrightarrow \overline{a(\pm \bar{w}_m)}=0, \hspace{0.5cm}m=1,...,N_2. \nonumber
\end{equation}
So the zeros of $a(z)$  come in pairs.
 It is convenient to define $\zeta_n=z_n$, $\zeta_{n+N_1}=-z_n$, $\zeta_{n+2N_1}=\bar{z}_n^{-1}$ and $\zeta_{n+3N_1}=-\bar{z}_n^{-1}$ for $n=1,\cdot\cdot\cdot,N_1$;  $\zeta_{m+4N_1}=w_m$ and $\zeta_{m+4N_1+N_2}=-w_m$ for $m=1,\cdot\cdot\cdot,N_2$. Therefore, the discrete spectrum is
\begin{equation}
	\mathcal{Z}=\left\{ \zeta_n, \  \bar{\zeta}_n\right\}_{n=1}^{4N_1+2N_2}, \label{spectrals}
\end{equation}
with $\zeta_n\in D^+$ and $\bar{\zeta}_n\in D^-$. And the distribution  of $	\mathcal{Z}$ on the $z$-plane   is shown  in Figure \ref{fig:figure1}.
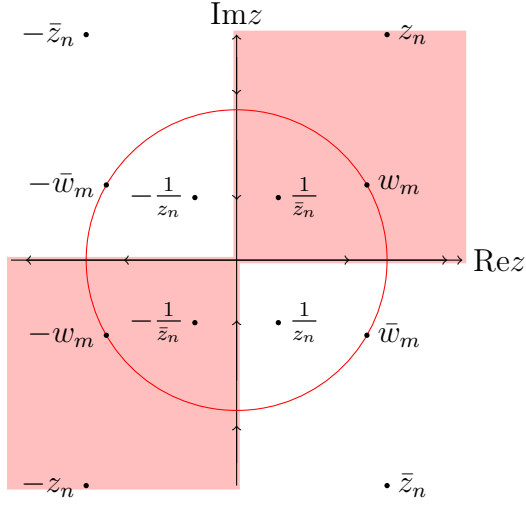
\begin{figure}[H]
	\centering
	\begin{tikzpicture}[node distance=2cm]
		\filldraw[pink,line width=3] (3,0.01) rectangle (0.01,3);
		\filldraw[pink,line width=3] (-3,-0.01) rectangle (-0.01,-3);
		\draw[->](-3,0)--(3,0)node[right]{Re$z$};
		\draw[->](0,-3)--(0,3)node[above]{Im$z$};
		\draw[red] (2,0) arc (0:360:2);
		\draw[->](0,0)--(-1.5,0);
		\draw[->](-1.5,0)--(-2.8,0);
		\draw[->](0,0)--(1.5,0);
		\draw[->](1.5,0)--(2.8,0);
		\draw[->](0,2.7)--(0,2.2);
		\draw[->](0,1.6)--(0,0.8);
		\draw[->](0,-2.7)--(0,-2.2);
		\draw[->](0,-1.6)--(0,-0.8);
			\coordinate (A) at (2,3);
		\coordinate (B) at (2,-3);
		\coordinate (C) at (-0.5546996232,0.8320505887);
		\coordinate (D) at (-0.5546996232,-0.8320505887);
		\coordinate (E) at (0.5546996232,0.8320505887);
		\coordinate (F) at (0.5546996232,-0.8320505887);
		\coordinate (G) at (-2,3);
		\coordinate (H) at (-2,-3);
		\coordinate (J) at (1.7320508075688774,1);
		\coordinate (K) at (1.7320508075688774,-1);
		\coordinate (L) at (-1.7320508075688774,1);
		\coordinate (M) at (-1.7320508075688774,-1);
		\fill (A) circle (1pt) node[right] {$z_n$};
		\fill (B) circle (1pt) node[right] {$\bar{z}_n$};
		\fill (C) circle (1pt) node[left] {$-\frac{1}{z_n}$};
		\fill (D) circle (1pt) node[left] {$-\frac{1}{\bar{z}_n}$};
		\fill (E) circle (1pt) node[right] {$\frac{1}{\bar{z}_n}$};
		\fill (F) circle (1pt) node[right] {$\frac{1}{z_n}$};
		\fill (G) circle (1pt) node[left] {$-\bar{z}_n$};
		\fill (H) circle (1pt) node[left] {$-z_n$};
		\fill (J) circle (1pt) node[right] {$w_m$};
		\fill (K) circle (1pt) node[right] {$\bar{w}_m$};
		\fill (L) circle (1pt) node[left] {$-\bar{w}_m$};
		\fill (M) circle (1pt) node[left] {$-w_m$};
	\end{tikzpicture}
	\caption{Distribution of the discrete spectrum $\mathcal{Z}$. The red one is  unit circle.}
	\label{fig:figure1}
\end{figure}

As shown in  \cite{CL2004},   denote   norming constant    $c_n=b_n/a'(z_n)$.    Then we have  residue conditions as
\begin{align}
 &\res_{z=\pm z_n}\left[\frac{\mu^1_{+}(z)}{a(z)}\right]=c_ne^{-2ik(\pm z_n)\lambda(\pm z_n)x}\mu^2_{-}(\pm z_n),\label{resrelation1}\\
	&\res_{z=\pm \bar{z}_n^{-1}}\left[\frac{\mu^1_{+}(z)}{a(z)}\right]=\pm \frac{\bar{q}_-}{q_-}\bar{z}_n^{-2}\bar{c}_ne^{-2ik(\pm \bar{z}_n^{-1})\lambda(\pm \bar{z}_n^{-1})x}\mu^2_{-}(\pm \bar{z}_n^{-1}),\\
	&\res_{z=\pm \bar{z}_n}\left[\frac{\mu^2_{+}(z)}{\overline{a(\bar{z})}}\right]=-\bar{c}_ne^{2ik(\pm \bar{z}_n)\lambda(\pm \bar{z}_n)x}\mu^1_{-}(\pm \bar{z}_n),\\
	&\res_{z=\pm z_n^{-1}}\left[\frac{\mu^2_{+}(z)}{\overline{a(\bar{z})}}\right]=\pm \frac{q_-}{\bar{q}_-}z_n^{-2}c_ne^{-2ik(\pm \bar{z}_n^{-1})\lambda(\pm z_n^{-1})x}\mu^2_{-}(\pm z_n^{-1}).
\end{align}

For $m=1,...,N_2$, there also have $c_{N_1+m}=b_{N_1+m}/a'(w_m)$ and
\begin{align}
	&\res_{z=\pm w_m}\left[\frac{\mu^1_{+}(z)}{a(z)}\right]=c_{N_1+m}e^{-2ik(\pm w_m)\lambda(\pm w_m)x}\mu^2_{-}(\pm w_m),\\
	&\res_{z=\pm \bar{w}_m}\left[\frac{\mu^2_{+}(z)}{\overline{a(\bar{z})}}\right]=-\bar{c}_{N_1+m}e^{2ik(\pm \bar{w}_m)\lambda(\pm \bar{w}_m)x}\mu^1_{-}(\pm \bar{w}_m).
\end{align}
For brevity, we introduce a new constant $C_n$ as: for  $n=1,...,N_1$, $C_n=C_{n+N_1}=c_n$, $C_{n+2N_1}=-C_{n+3N_1}=\frac{\bar{q}_-}{q_-}\bar{z}_n^{-2}\bar{c}_n$; for  $m=1,...,N_2$, $C_{m+4N_1}=C_{m+4N_1+N_2}=c_{m+N_1}$,
and  the collection
$\sigma_d=  \left\lbrace \zeta_n,C_n\right\rbrace^{4N_1+2N_2}_{n=1}  $
is called the \emph{scattering data}.

Now we are going to take into account the time evolution of scattering data. If $q$ also depends on t (i.e. $q$ = $q(x, t))$, we can obtain the functions $a$ and $b$ as above for all times $t \in R$. Taking account of the t-part in (\ref{lax0}), the t- derivative of $a$ and $b$ comes to
\begin{align}
	  a_t(z;t) =0,\hspace{0.5cm}  b_t(z;t) =-(2k^2-1)k\lambda b(z;t).
\end{align}
Then we can obtain time dependence of scattering data which can be expressed as the following replacement
\begin{align}
	&C(\zeta_n)\rightarrow C(t,\zeta_n)=c(0,\zeta_n)e^{-(2k(\zeta_n)^2-1)k(\zeta_n)\lambda(\zeta_n) t},\\
	&r(z)\rightarrow r(t,z)=r(0,z)e^{-(2k^2-1)k\lambda t}
\end{align}
In particular, if at time $t = 0$ the initial function $q(x, 0)$ produces $4N_1+2N_2$ simple zeros $\zeta_1$,...,$\zeta_{4N_1+2N_2}$ of $a(z; 0)$ and if $q$ evolves accordingly
 to the (\ref{DNLS}), then $q(x, t)$ will produce exactly the same N simple zeros at any other time $t \in   R$. The scattering data with time $t$   is given by
\begin{equation*}
	\left\lbrace  e^{-(2k^2-1)k\lambda t}r(z),\left\lbrace \zeta_n,e^{-(2k(\zeta_n)^2-1)k(\zeta_n)\lambda(\zeta_n) t}C_n\right\rbrace^{4N_1+2N_2}_{n=1}\right\rbrace,
\end{equation*}
where $\left\lbrace  r(z),\left\lbrace \zeta_n,C_n\right\rbrace^{4N_1+2N_2}_{n=1}\right\rbrace$ are obtained from the initial data $q(x, 0) =
q_0(x)$. Denote the  phase function
\begin{equation}
	\theta(z)=k(z)\lambda(z)\left[x/t-(2k(z)^2-1) \right],\label{theta}
\end{equation}
and for convenience we denote $\theta_n=\theta(\zeta_n)$.

To propose and solve the matrix RH problem in the following inverse problem, we finally give the asymptotic behaviors of the modified Jost solutions
and scattering matrix as $z\to\infty$ and $z\to 0$.
\begin{Proposition} \label{prop2}
	The Jost solutions posses  the following asymptotic behaviors
	\begin{align}
		&\mu_\pm(x,t,z)=e^{i\nu_\pm(x,t;q)\sigma_3}+O(z^{-1}),\hspace{0.5cm}z \rightarrow \infty,\label{asyvarphi1}\\
		&\mu_\pm(x,t,z)=\frac{i}{z}e^{i\nu_\pm(x,t;q)\sigma_3}\sigma_3Q_\pm+O(1),\hspace{0.5cm}z \rightarrow 0\label{asyvarphi2},
	\end{align}
	where
	\begin{equation}
		\nu_\pm(x,t;q)=\frac{1}{2}\int_{\pm\infty}^x (|q|^2-1)dy.
	\end{equation}
	The  scattering matrices admit   asymptotic behaviors
	\begin{align}
		&S(z)=e^{-i\nu_0(t;q)\sigma_3}+O(z^{-1}),\hspace{0.5cm}z \rightarrow \infty,\label{asympsca1}\\
		&S(z)={\rm diag}\left(\frac{q_-}{q_+},\frac{q_+}{q_-}\right)e^{i\nu_0(t;q)\sigma_3}+O(z),\hspace{0.5cm}z \rightarrow 0,\label{asympsca2}
	\end{align}
	where
	\begin{equation}
		\nu_0(t;q)=\frac{1}{2}\int_{-\infty}^{+\infty}(|q|^2-1)dy.\label{nu0}
	\end{equation}
	Further we have $\rho(0)=\tilde{\rho}(0)=0$.
\end{Proposition}
Moreover, from trace formulae we have
\begin{equation}
	a(z)=\prod_{j=1}^{4N_1+2N_2}\frac{z-\zeta_j}{z-\bar{\zeta}_j}\exp\left\lbrace-\frac{1}{2\pi i}\int_{\Sigma}\frac{\log (1-\rho(s)\tilde{\rho}(s))}{s-z}ds \right\rbrace .\label{a}
\end{equation}
Then by taking $z\to 0$,  theta condition is obtained:
\begin{align}
	\arg\frac{q_-}{q_+}+2\nu_0=8\sum_{n=1}^{N_1}\arg(z_n)+4\sum_{m=1}^{N_2}\arg(w_m)+\frac{1}{2\pi}\int_{\Sigma}\frac{\log(1-\rho(s)\tilde{\rho}(s))}{s}ds+2j\pi, \hspace{0.5cm}
\end{align}
where $j$  is a  integer.

Define  a   sectionally meromorphic matrix
\begin{equation}
M(z;x,t)=\left\{ \begin{array}{ll}
\left( a(z)^{-1}\mu_+ ^1, \mu_-^2\right),   &\text{as } z\in D^+,\\[12pt]
\left( \mu_-^1,\overline{a(\bar{z})}^{-1}\mu_+^2\right)  , &\text{as }z\in D^-,\\
\end{array}\right.
\end{equation}
which solves the following (time-dependent) RHP.

\noindent\textbf{RHP0}.  Find a matrix-valued function $M(z)$ which satisfies:

$\blacktriangleright$ Analyticity: $M(z)$ is meromorphic in $\mathbb{C}\setminus \Sigma$ and has single poles
$\mathcal{Z}$;

$\blacktriangleright$ Symmetry: $M(z)=\sigma_2\overline{M(\bar{z})}\sigma_2$=$\sigma_1\overline{M(-\bar{z})}\sigma_1=\frac{i}{z}M(-1/z)\sigma_3Q_-$;

$\blacktriangleright$ Jump condition: $M(z)$ has continuous boundary values $M_\pm(z)$ on $\Sigma$ and
\begin{equation}
M_+(z )=M_-(z )V(z),\hspace{0.5cm}z \in \Sigma,
\end{equation}
where
\begin{equation}
V(z)=\left(\begin{array}{cc}
1-\tilde{\rho}(z)\rho(z) & -e^{2it\theta}\tilde{\rho}(z)\\
e^{-2it\theta}\rho(z) & 1
\end{array}\right);\label{jumpv}
\end{equation}

$\blacktriangleright$ Asymptotic behaviors:
\begin{align}
&M(z ) = e^{i\nu_-(x,t;q)\sigma_3}+\mathcal{O}(z^{-1}),\hspace{0.5cm}z \rightarrow \infty,\\
&M(z ) =\frac{i}{z}e^{i\nu_-(x,t;q)\sigma_3}\sigma_3Q_-+\mathcal{O}(1),\hspace{0.5cm}z \rightarrow 0;
\end{align}

$\blacktriangleright$ Residue conditions: $M$ has simple poles at each point in $ \mathcal{Z}\cup \bar{\mathcal{Z}}$ with:
\begin{align}
&\res_{z=\zeta_n}M(z)=\lim_{z\to \zeta_n}M(z)\left(\begin{array}{cc}
0 & 0\\
C_ne^{-2it\theta_n} & 0
\end{array}\right),\label{RES1}\\
&\res_{z=\bar{\zeta}_n}M(z)=\lim_{z\to \bar{\zeta}_n}M(z)\left(\begin{array}{cc}
0 & -\bar{C}_ne^{2it\bar{\theta}_n}\\
0 & 0
\end{array}\right)\label{RES2}.
\end{align}

From the asymptotic behavior in Proposition \ref{proasyM}, the reconstruction formula of $q(x,t)$ is given by
\begin{equation}
q(x,t)=\exp\left\lbrace \frac{i}{2}\int_{-\infty}^x(|q(x,t)|^2-1)dy \right\rbrace m(x,t),\label{recons q}
\end{equation}
where
\begin{align}
m(x,t)=\lim_{z\to \infty}\left[zM \right]_{12}.\label{recons u}
\end{align}
Take modulus on both sides of (\ref{recons q}) yields
\begin{align}
|q(x,t)|=|m(x,t)|,\nonumber
\end{align}
which is substituted back into  (\ref{recons q}) leads to
\begin{equation}
	q(x,t)=\exp\left\lbrace \frac{i}{2}\int_{-\infty}^x(|m(x,t)|^2-1)dy \right\rbrace m(x,t).\label{q}
\end{equation}

\section{Deformation to a mixed $\bar\partial$-RH problem}\label{sec3}

\quad
We find that  the long-time asymptotic  of RHP0  is affected by the growth and decay of the exponential function $e^{\pm2it\theta}$ appearing in both the jump relation and the residue conditions.
Therefore, in this section, we introduce  a new transform  $M(z)\to M^{(1)}(z)$,  which  make that the  $M^{(1)}(z)$ is well behaved as $t\to \infty$ along any characteristic line.

Let $\xi=\frac{x}{t}$,  to obtain asymptotic behavior  of $e^{2it\theta}$ as $t\to \infty$, we consider the real part of $2it\theta$:
\begin{equation}
\text{Re}(2it\theta)=-t\text{Im}z\text{Re}z\left[ \left( \xi+2\right)\left(1+|z|^{-4} \right) -\left( \text{Re}^2z-\text{Im}^2z\right)\left( 1+|z|^{-8}\right)   \right]  .\label{Reitheta}
\end{equation}
The signature  of $\text{Im}\theta$ are shown in Figure \ref{figtheta}.

\begin{figure}[H]
	\centering
	\subfigure[]{
		\includegraphics[width=0.25\linewidth]{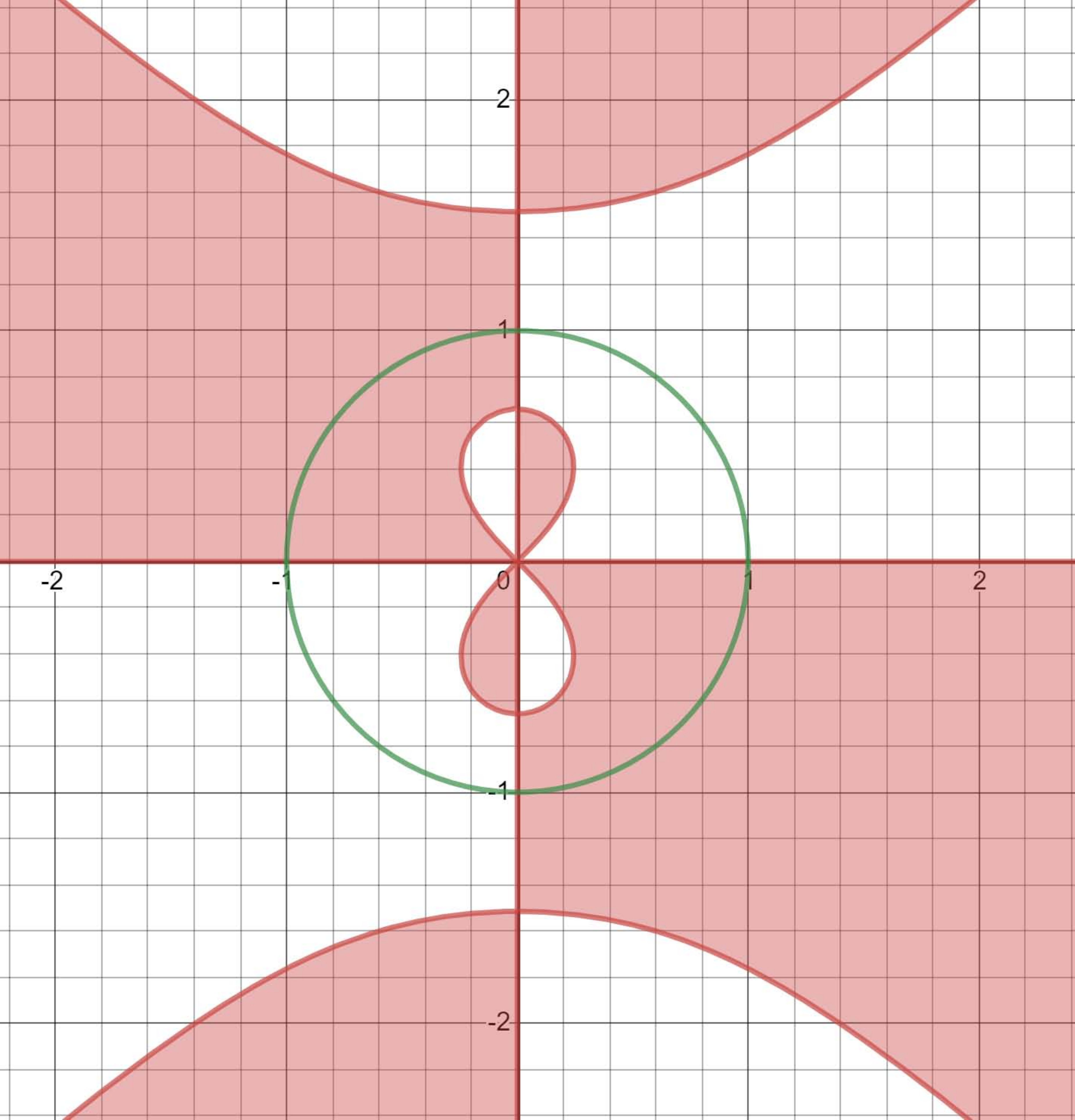}
	\label{fig:desmos-graph}}
	\subfigure[]{
		\includegraphics[width=0.25\linewidth]{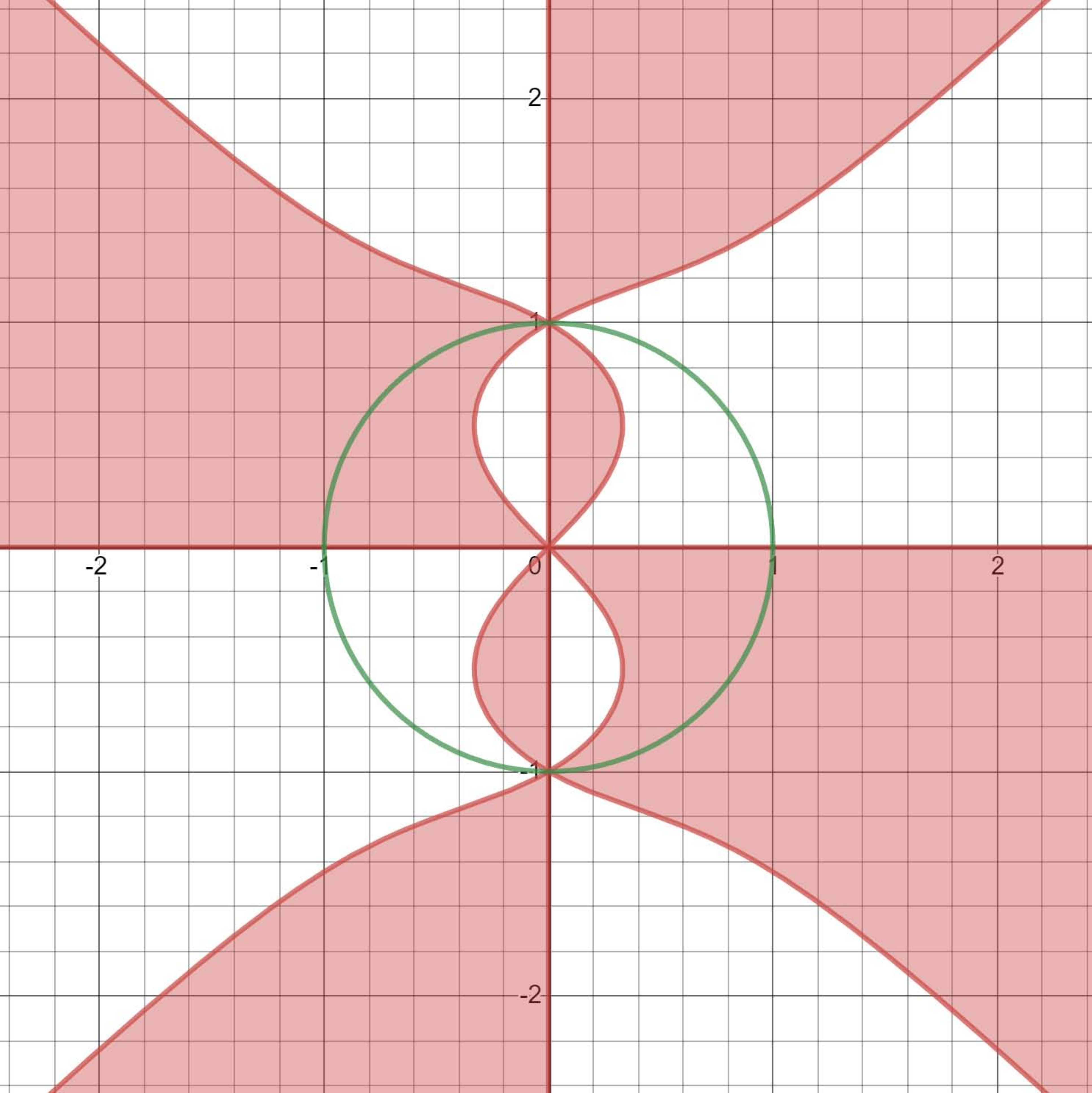}
		\label{fig:desmos-graph-1}}
	\subfigure[]{
		\includegraphics[width=0.251\linewidth]{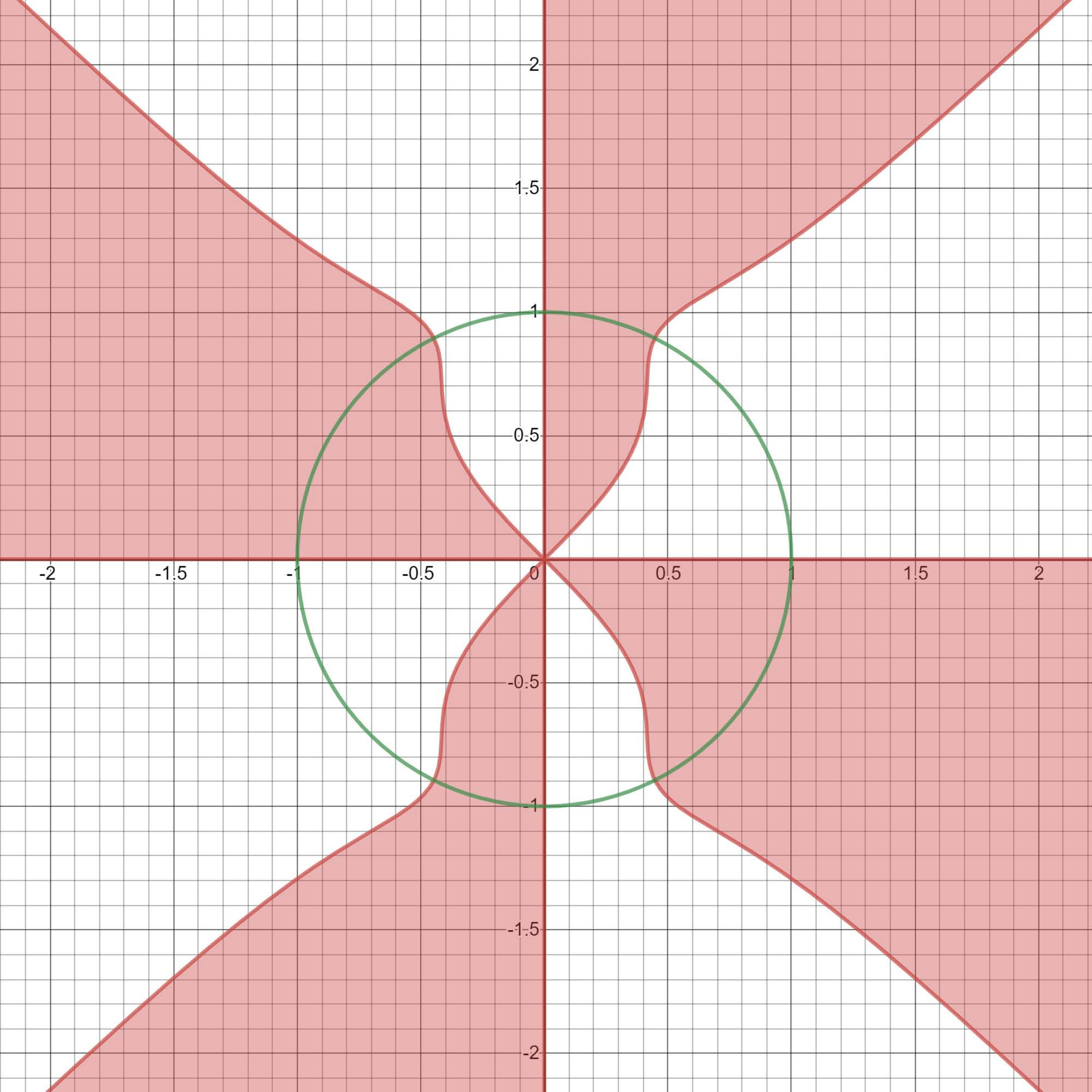}
		\label{fig:desmos-graph-5}}
	\subfigure[]{
		\includegraphics[width=0.25\linewidth]{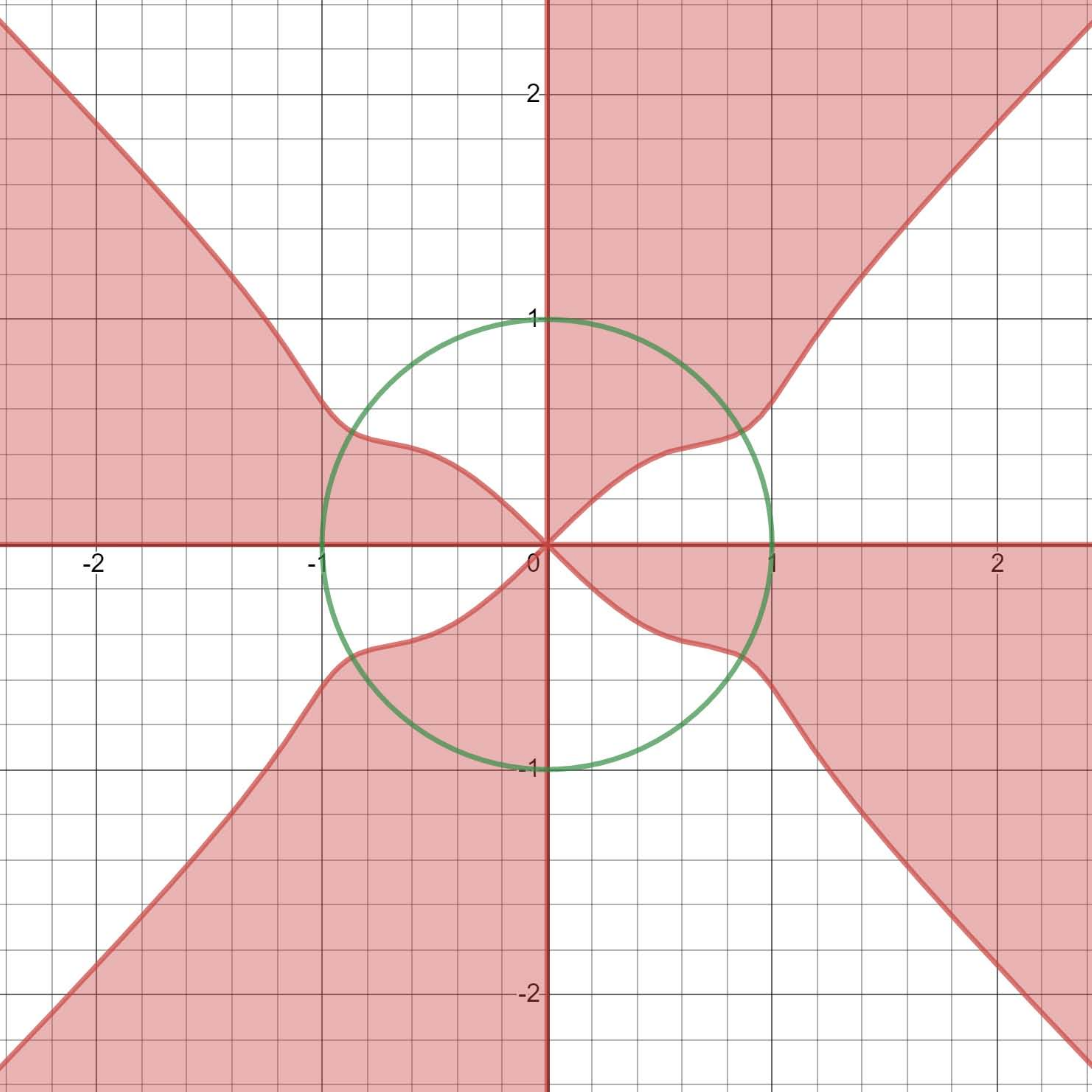}
		\label{fig:desmos-graph-2}}
	\subfigure[]{
		\includegraphics[width=0.25\linewidth]{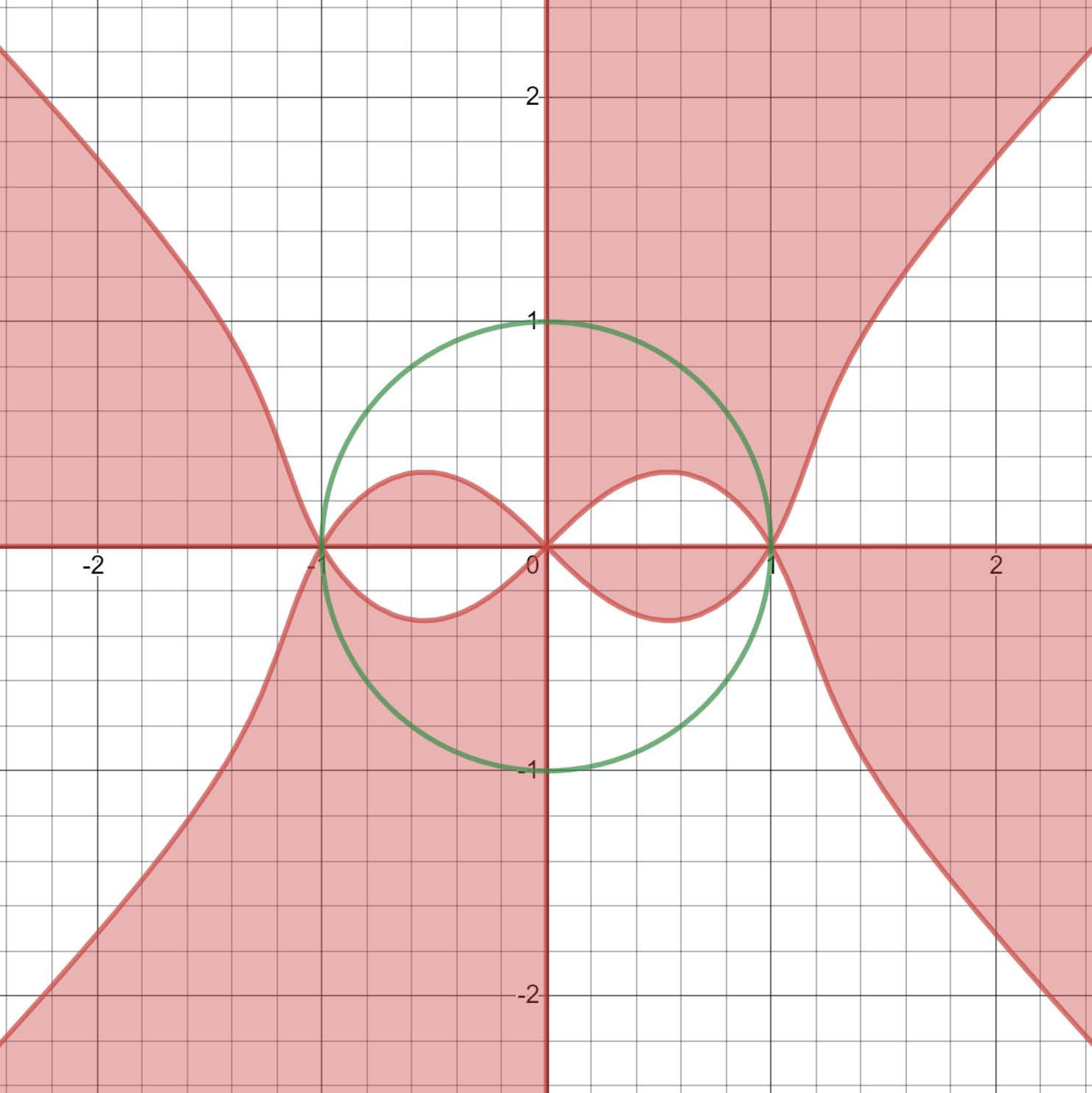}
		\label{fig:desmos-graph-3}}
	\subfigure[]{
		\includegraphics[width=0.25\linewidth]{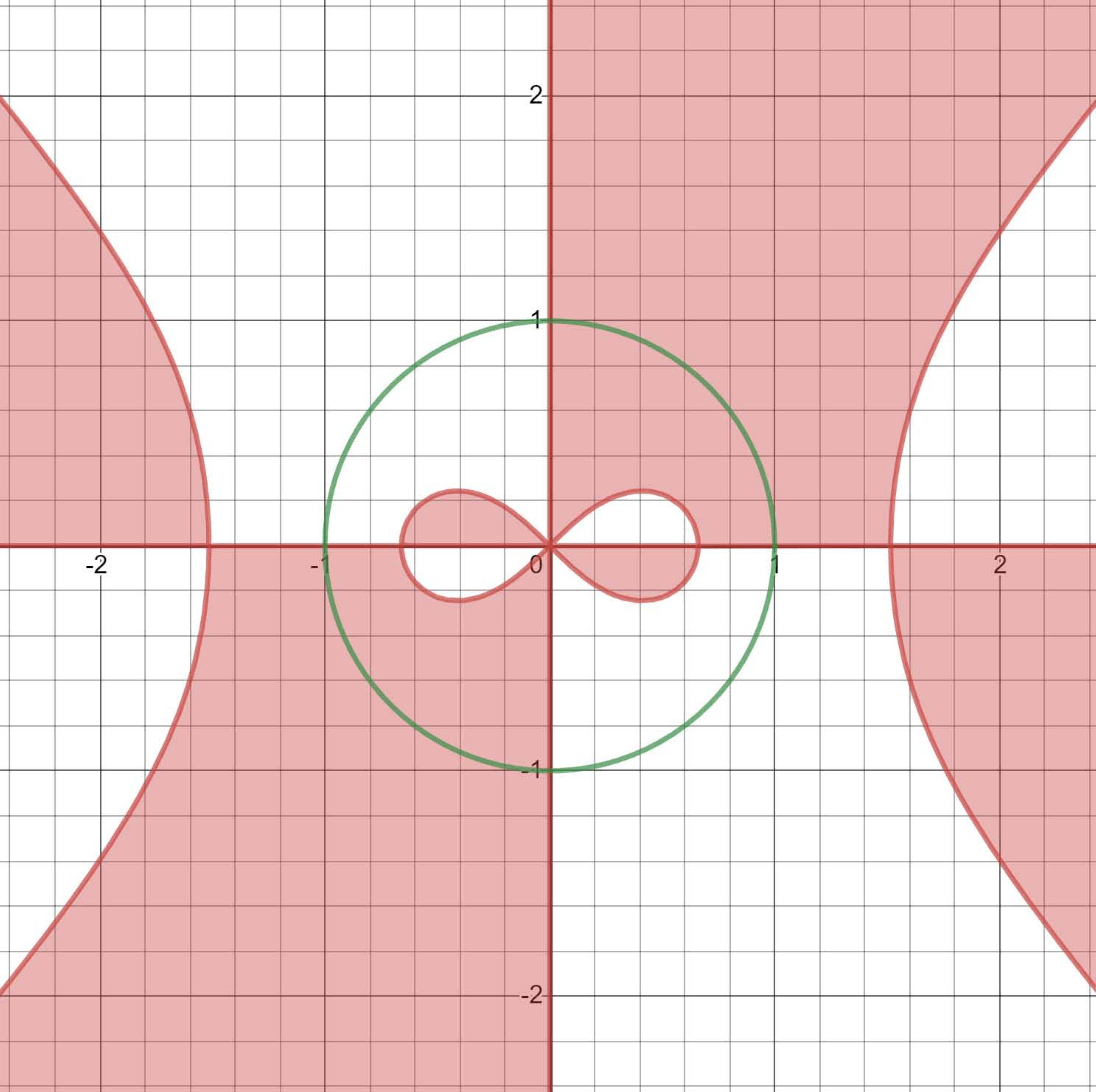}
		\label{fig:desmos-graph-4}}
\caption{In these figure we take $\xi=-4,-3,-2.6,-1.5,-1,0$
respectively to show all type of $\text{Im}\theta$. The green curve is  unit circle. In the red region,
$\text{Im}\theta>0$ while $\text{Im}\theta=0$ on the red  curve. And $\text{Im}\theta<0$ in the white region.}
\label{figtheta}
\end{figure}

In our paper, we only consider the case $-3<\xi<-1$ which is corresponding to Figure \ref{figtheta} (c), because of its well property.

 For brevity, we introduce some notations with respect to subscripts
\begin{align}
&\mathcal{N}\triangleq\left\lbrace 1,...,4N_1+2N_2\right\rbrace, \ \ \nabla=\left\lbrace n \in  \mathcal{N}  |\text{Im}\theta_n\leq 0\right\rbrace,\nonumber\\
&\Delta=\left\lbrace n \in  \mathcal{N} |\text{Im}\theta_n> 0\right\rbrace,\Lambda=\left\lbrace n \in  \mathcal{N}  |\text{Im}\theta_n= 0\right\rbrace.\label{devide}
\end{align}
For   $n\in\Delta$, the residue of $M(z)$ at $\zeta_n$ in (\ref{RES1})  are unbounded   as $t\to\infty$.
Similarly, for   $n\in\nabla$, the residue at $\zeta_n$   approach   to be zero as $t\to\infty$.
 Define
\begin{equation}
	\rho_0=\min_{n\in\Delta\cup\nabla\setminus\Lambda}|\text{Im}\theta_n|\neq0.\label{rho0}
\end{equation}
To distinguish different  type of zeros, we further give
\begin{align}
	&\nabla_1=\left\lbrace j \in \left\lbrace 1,...,N_1\right\rbrace  |\text{Im}\theta(z_j)\leq 0\right\rbrace,
	\Delta_1=\left\lbrace j \in \left\lbrace 1,...,N_1\right\rbrace  |\text{Im}\theta(z_j)> 0\right\rbrace,\nonumber\\
	&\nabla_2=\left\lbrace i \in \left\lbrace 1,...,N_2\right\rbrace  |\text{Im}\theta(w_i)\leq 0\right\rbrace,
	\Delta_2=\left\lbrace i \in \left\lbrace 1,...,N_2\right\rbrace  |\text{Im}\theta(w_i)> 0\right\rbrace,\nonumber\\
	&\Lambda_1=\left\lbrace j_0 \in \left\lbrace 1,...,N_1\right\rbrace  |\text{Im}\theta(z_{j_0})= 0\right\rbrace,\Lambda_2=\left\lbrace i_0 \in \left\lbrace 1,...,N_2\right\rbrace  |\text{Im}\theta(w_{i_0)}= 0\right\rbrace.\nonumber
\end{align}

For the poles $\zeta_n$ with $n\notin\Lambda$, we want to trap them for jumps along small closed circles
enclosing themselves respectively. The jump matrix  in  (\ref{jumpv})  also needs to  be restricted.  Recall the  well known factorizations of  $V(z)$:
\begin{align}
	V(z)&=\left(\begin{array}{cc}
	1 & -\tilde{\rho} e^{2it\theta}\\
	0 & 1
\end{array}\right)\left(\begin{array}{cc}
1 & 0 \\
\rho e^{-2it\theta} & 1
\end{array}\right)\\
&=\left(\begin{array}{cc}
	1 & \\
	\frac{\rho e^{-2it\theta}}{1-\rho\tilde{\rho}} & 1
\end{array}\right)(1-\rho\tilde{\rho})^{\sigma_3}\left(\begin{array}{cc}
1 & -\frac{\tilde{\rho} e^{2it\theta}}{1-\rho\tilde{\rho}} \\
0 & 1
\end{array}\right).
\end{align}
We will use  these factorizations to deform the jump contours so that exponentials $e^{\pm2it\theta}$ are decaying in corresponding regions  respectively.
Define  functions
\begin{align}
\delta (z)&=\exp\left(-\frac{1}{2\pi i}\int _{i\mathbb{R}}\left( \dfrac{1}{s-z}-\frac{1}{2s}\right) \log (1-\rho(s)\tilde{\rho}(s))ds\right);\\
T(z)&=T(z,\xi)=\prod_{n\in \Delta}\dfrac{z-\zeta_n}{\bar{\zeta}_n^{-1}z-1}\delta (z)\nonumber\\
&=\prod_{j\in \Delta_1}\dfrac{z^2-z_j^2}{\bar{z}_j^{-2}z^2-1}\dfrac{z^2-\bar{z}_j^{-2}}{z_j^{2}z^2-1}\prod_{i\in \Delta_2}\dfrac{z^2-w_i^2}{w_i^{2}z^2-1}\delta (z) \label{T}.
\end{align}
In  the above formulas, we choose the principal branch of power and logarithm functions.

\begin{Proposition}\label{proT}
	The function defined by (\ref{T}) has following properties:\\
	(a) $T$ is meromorphic in $\mathbb{C}\setminus \mathbb{R}$, and for each $n\in\Delta$, $T(z)$ has simple  zeros   $\zeta_n$ and  simple poles  $\bar{\zeta}_n$;\\
	(b) $T(z)=\overline{T^{-1}(\bar{z})}=T^{-1}(-z^{-1})$;\\
	(c) For $z\in \mathbb{R}$, as z approaching the real axis from above and below, $T$ has boundary values $T_\pm$, which satisfy:
	\begin{equation}
	T_+(z)=(1-\rho(z)\tilde{\rho}(z))T_-(z),\hspace{0.5cm}z\in i\mathbb{R};
	\end{equation}
	(d)  $\lim_{z\to \infty}T(z)\triangleq T(\infty)$, where
	\begin{equation}
		T(\infty)=\prod_{j\in \Delta_1} \bar{z}_j^2z_j^{-2}\prod_{i\in \Delta_2}\bar{w}_i^{2}\exp\left(\frac{1}{4\pi i}\int_{i\mathbb{R}}s^{-1}  \log (1-\rho(s)\tilde{\rho}(s))ds\right) ,
	\end{equation}
	with $ |T(\infty)|=1$;\\
	(e)  As $|z|\to \infty$ with $|arg(z)|\leq c<\pi$,
	\begin{equation}
	T(z)=T(\infty)\left( 1+z^{-1}\frac{1}{2\pi i}\int _{i\mathbb{R}}\log (1-\rho(s)\tilde{\rho}(s))ds+ \mathcal{O}(z^{-2})\right) \label{expT};
	\end{equation}
	(f)  $T(z)$ is continuous at $z=0$, and
	\begin{equation}
	\lim_{z\to 0}T(z)=T(0)=T(\infty)^{-1} \label{T0};
	\end{equation}
	(g)  $\frac{a(z)}{T(z)}$ is holomorphic in $D^+$. And its  absolute value  is bounded in $D^+\cap \left\lbrace z\in\mathbb{C}|\text{Re} z>0 \right\rbrace $. Additionally, the ratio extends as a continuous function on $i\mathbb{R}$.
\end{Proposition}
\begin{proof}
	 Properties (a), (b), (d) and (f)  can be obtain by simple calculation. And (c)  follows from the Plemelj formula. By the Laurent expansion (e)  immediately. For brevity, we  omit calculation. For (g), from (\ref{a}) we have
	 \begin{equation}
	 	\frac{a(z)}{T(z)}=T(\infty)^{-1}\prod_{j\in \nabla_1}\dfrac{z^2-z_j^2}{\bar{z}_j^{-2}z^2-1}\dfrac{z^2-\bar{z}_j^{-2}}{z_j^{2}z^2-1}\prod_{i\in \nabla_2}\dfrac{z^2-w_i^2}
{w_i^{2}z^2-1}\exp\left\lbrace -\frac{1}{2\pi i}\int _{\mathbb{R}}\frac{\log (1-\rho(s)\tilde{\rho}(s))}{s-z}\right\rbrace.\nonumber
	 \end{equation}
 So $\frac{a(z)}{T(z)}$ is holomorphic in $D^+$. And in above expression,  all factors except the last  integral is bounded for $z\in D^+$. From (\ref{symrho}), $1-\rho(s)\tilde{\rho}(s)=1+|\rho(s)|^2$. Let $z=x+yi$, then  the
 real part of the exponential
 is $		-\frac{y}{2\pi}\int _{\mathbb{R}}\frac{\log (1+|\rho(s)|^2)}{|s-z|^2}ds$ which can be bounded as follows:
	 \begin{align}
|\frac{y}{2\pi}\int _{\mathbb{R}}\frac{\log (1+|\rho(s)|^2)}{|s-z|^2}ds|&\leq\frac{1}{2\pi}\parallel
\log (1+|\rho(s)|^2)\parallel_{L^\infty(\mathbb{R})}\parallel \frac{y}{(s-x)^2+y^2}\parallel_{L^1(\mathbb{R})}\nonumber\\
&\lesssim \parallel \rho(s)\parallel_{L^\infty(\mathbb{R})}.\nonumber
	\end{align}

\end{proof}

Additionally, let $\varrho$ be a positive constant stratifying
\begin{equation}
	\varrho=\frac{1}{2}\min\left\lbrace\min_{ j\neq i\in \mathcal{N}}|\zeta_i-\zeta_j|, \min_{j\in \mathcal{N}}\left\lbrace |\text{Im}\zeta_j|, |\text{Re}\zeta_j|\right\rbrace ,\min_{j\in \mathcal{N}\setminus\Lambda,\text{Im}\theta(z)=0}|\zeta_j-z| \right\rbrace .
\end{equation}
By above definition, for every $n\in  \mathcal{N}  $, we define disks   $\mathbb{D}(\zeta_n,\varrho)$, such that  they   pairwise disjoint,
also    disjoint  with $\left\lbrace z\in \mathbb{C}|\text{Im} \theta(z)=0 \right\rbrace $ and $\Sigma$.
Introduce a piecewise matrix function
\begin{equation}
	G(z)=\left\{ \begin{array}{ll}
		\left(\begin{array}{cc}
			1 & 0\\
			-C_n(z-\zeta_n)^{-1}e^{-2it\theta_n} & 1
		\end{array}\right),   &\text{as } z\in\mathbb{D}(\zeta_n,\varrho),n\in\nabla\setminus\Lambda;\\[12pt]
		\left(\begin{array}{cc}
			1 & -C_n^{-1}(z-\zeta_n)e^{2it\theta_n}\\
			0 & 1
		\end{array}\right),   &\text{as } z\in\mathbb{D}(\zeta_n,\varrho),n\in\Delta;\\
		\left(\begin{array}{cc}
		1 & \bar{C}_n(z-\bar{\zeta}_n)^{-1}e^{2it\bar{\theta}_n}\\
		0 & 1
		\end{array}\right),   &\text{as } 	z\in\mathbb{D}(\bar{\zeta}_n,\varrho),n\in\nabla\setminus\Lambda;\\
		\left(\begin{array}{cc}
		1 & 0	\\
		\bar{C}_n^{-1}(z-\bar{\zeta}_n)e^{-2it\bar{\theta}_n} & 1
		\end{array}\right),   &\text{as } 	z\in\mathbb{D}(\bar{\zeta}_n,\varrho),n\in\Delta;\\
	I &\text{as } 	z \text{ in elsewhere};
	\end{array}\right..\label{funcG}
\end{equation}
Now we use $T(z)$ and $G(z)$  to define a new  matrix-valued   function $M^{(1)}(z)$.
\begin{equation}
M^{(1)}(z)=T(\infty)^{-\sigma_3}M(z)G(z)T(z)^{\sigma_3},\label{transm1}
\end{equation}
which then satisfies the following RH problem.

\noindent \textbf{RHP1}. Find a matrix-valued function  $  M^{(1)}(z )$ which satisfies:

$\blacktriangleright$ Analyticity: $M^{(1)}(z )$ is meromorphic in $\mathbb{C}\setminus \Sigma^{(1)}$, where
\begin{equation}
	\Sigma^{(1)}=\mathbb{R}\cup i\mathbb{R}\cup\left[\cup_{n\in\mathcal{N}\setminus\Lambda}\left( \partial \mathbb{D}(\bar{\zeta}_n,\varrho)\cup \partial\mathbb{D}(\zeta_n,\varrho)\right)  \right] ,
\end{equation}
is shown in Figure \ref{fig:zero};

$\blacktriangleright$ Symmetry: $M^{(1)}(z)=\sigma_2\overline{M^{(1)}(\bar{z})}\sigma_2$=$\sigma_1\overline{M^{(1)}(-\bar{z})}\sigma_1=\frac{i}{z}M^{(1)}(-1/z)\sigma_3Q_-$;

$\blacktriangleright$ Jump condition: $M^{(1)}$ has continuous boundary values $M^{(1)}_\pm$ on $\Sigma^{(1)}$ and
\begin{equation}
	M^{(1)}_+(z)=M^{(1)}_-(z)V^{(1)}(z),\hspace{0.5cm}z \in \Sigma^{(1)},
\end{equation}
where
\begin{equation}
	V^{(1)}(z)=\left\{\begin{array}{ll}\left(\begin{array}{cc}
		1 & -e^{2it\theta}\tilde{\rho}(z)T^{-2}(z) \\
		0 & 1
	\end{array}\right)
		\left(\begin{array}{cc}
			1 & 0\\
			e^{-2it\theta}\rho(z)T^2(z) & 1
		\end{array}\right),   &\text{as } z\in 	\mathbb{R};\\[12pt]
		\left(\begin{array}{cc}
		1 & 0\\
		\frac{e^{-2it\theta}\tilde{\rho}(z)T_+^{2}(z)}{1-\tilde{\rho}(z)\rho(z)} & 1
	\end{array}\right)\left(\begin{array}{cc}
	1 & -\frac{e^{2it\theta}\tilde{\rho}(z)T_-^{-2}(z)}{1-\tilde{\rho}(z)\rho(z)}\\
	0 & 1
\end{array}\right),   &\text{as } z\in i\mathbb{R};\\[12pt]
		\left(\begin{array}{cc}
		1 & 0\\
		-C_n(z-\zeta_n)^{-1}T^2(z)e^{-2it\theta_n} & 1
		\end{array}\right),   &\text{as } 	z\in\partial\mathbb{D}(\zeta_n,\varrho),n\in\nabla\setminus\Lambda;\\[12pt]
		\left(\begin{array}{cc}
			1 & -C_n^{-1}(z-\zeta_n)T^{-2}(z)e^{2it\theta_n}\\
			0 & 1
		\end{array}\right),   &\text{as } z\in\partial\mathbb{D}(\zeta_n,\varrho),n\in\Delta;\\
		\left(\begin{array}{cc}
			1 & \bar{C}_n(z-\bar{\zeta}_n)^{-1}T^{-2}(z)e^{2it\bar{\theta}_n}\\
			0 & 1
		\end{array}\right),   &\text{as } 	z\in\partial\mathbb{D}(\bar{\zeta}_n,\varrho),n\in\nabla\setminus\Lambda;\\
		\left(\begin{array}{cc}
			1 & 0	\\
			\bar{C}_n^{-1}(z-\bar{\zeta}_n)e^{-2it\bar{\theta}_n}T^2(z) & 1
		\end{array}\right),   &\text{as } 	z\in\partial\mathbb{D}(\bar{\zeta}_n,\varrho),n\in\Delta;\\
	\end{array}\right.;\label{jumpv1}
\end{equation}

$\blacktriangleright$ Asymptotic behaviors:
\begin{align}
	&M^{(1)}(z) = e^{i\nu_-(x,t;q)\sigma_3}+\mathcal{O}(z^{-1}),\hspace{0.5cm}z \rightarrow \infty,\\
	&M^{(1)}(z) =\frac{i}{z}e^{i\nu_-(x,t;q)\sigma_3}\sigma_3Q_-+\mathcal{O}(1),\hspace{0.5cm}z \rightarrow 0;
\end{align}

$\blacktriangleright$ Residue conditions: $M^{(1)}$ has simple poles at each point $\zeta_n$ and $\bar{\zeta}_n$ for $n\in\Lambda$ with:
\begin{align}
	&\res_{z=\zeta_n}M^{(1)}(z)=\lim_{z\to \zeta_n}M^{(1)}(z)\left(\begin{array}{cc}
		0 & 0\\
		C_ne^{-2it\theta_n}T^2(\zeta_n) & 0
	\end{array}\right),\\
	&\res_{z=\bar{\zeta}_n}M^{(1)}(z)=\lim_{z\to \bar{\zeta}_n}M^{(1)}(z)\left(\begin{array}{cc}
		0 & -\bar{C}_nT^{-2}(\bar{\zeta}_n)e^{2it\bar{\theta}_n}\\
		0 & 0
	\end{array}\right).
\end{align}

\begin{proof}
	Note that the triangular factors  (\ref{funcG})  trades 	poles $\zeta_n$  and $\bar{\zeta}_n$  to jumps on the disk boundaries $\partial \mathbb{D}(\zeta_n,\varrho)$ and $\partial \mathbb{D}(\bar{\zeta}_n,\varrho)$ respectively  for $n\in\mathcal{N}\setminus\Lambda$. Then by simple calculation we can obtain the residues condition and jump condition from (\ref{RES1}), (\ref{RES2}) (\ref{jumpv}), (\ref{funcG}) and (\ref{transm1}). The   analyticity and symmetry of $M^{(1)}(z)$ is directly from its definition, the Proposition \ref{proT}, (\ref{funcG}) and the properties of $M$. As for asymptotic behaviors, from $\lim_{z\to 0}G(z)=\lim_{z\to \infty}G(z)=I$ and Proposition \ref{proT} (f), we  obtain that $M^{(1)}(z)$ has same asymptotic behaviors as $M(z)$.
\end{proof}

\begin{figure}[H]
	\centering
	\subfigure[]{
		\includegraphics[width=0.5\linewidth]{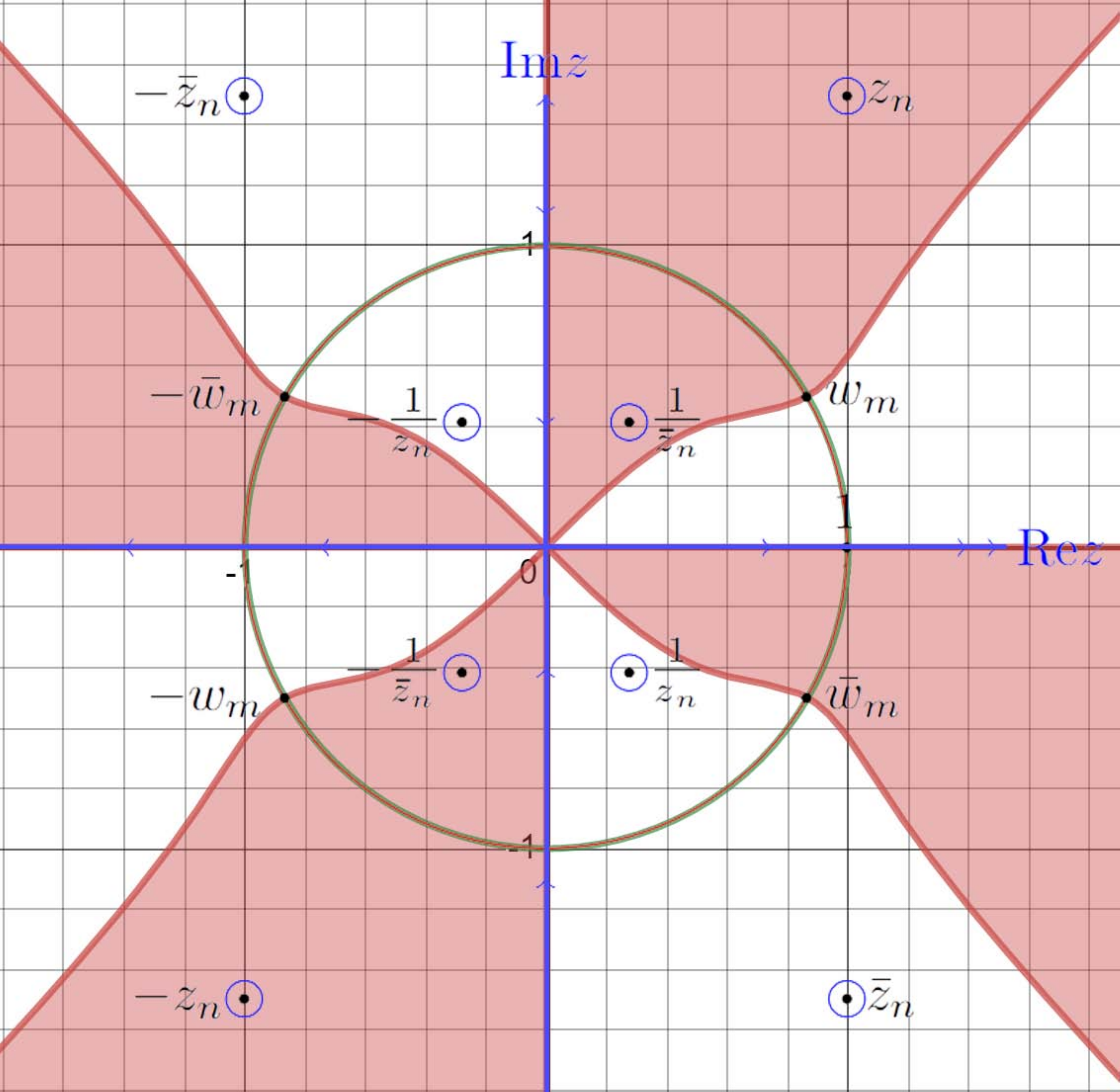}
	\label{fig:desmos-graph}}
\caption{The blue curve, including $\mathbb{R}$, $i\mathbb{R}$ and the small circles  constitute $\Sigma^{(1)}$.
 Because Im$\theta(w_m)=0$, it remain the pole of $M^{(1)}$. And Im$\theta(z_n)\neq0$, so we change it to
  jump on $\partial\mathbb{D}(\zeta_n,\varrho)$. }	\label{fig:zero}
\end{figure}

\section{Mixed $\bar{\partial}$-RH Problem }\label{sec4}

\quad  In this section,  we make continuous extension for  the jump matrix $V^{(1)}$  to remove the jump from $\Sigma$. Besides, the new problem is hoped to  takes advantage of the decay/growth of $e^{2it\theta(z)}$ for $z\notin\Sigma$. For this purpose, we   introduce new eight regions:
\begin{align}
	&\Omega_{2n+1}=\left\lbrace z\in\mathbb{C}|n\pi/2 \leq\arg z \leq n\pi/2+\varphi \right\rbrace ,\\
	&\Omega_{2n+2}=\left\lbrace z\in\mathbb{C}|(n+1)\pi/2 -\varphi\leq\arg z \leq (n+1)\pi/2 \right\rbrace,
\end{align}
where $n=0,1,2,3$ and $\varphi>0$ is an fixed sufficiently small angle  achieving following conditions:\\
1.  $\frac{2|\xi+2|}{|\xi+2|+1}<\cos 2\varphi<1$;\\
2.  each $\Omega_i$ doesn't intersect any of $\mathbb{D}(\zeta_n,\varrho)$ or $\mathbb{D}(\bar{\zeta}_n,\varrho)$.\\
Define  new contours as follow:
\begin{align}
&\Sigma_k=e^{(k-1)i\pi/4+\varphi}R_+,\hspace{0.5cm}k=1,3,5,7,\\
&\Sigma_k=e^{ki\pi/4-\varphi}R_+,\hspace{0.5cm}k=2,4,6,8,\\
&\tilde{\Sigma}=\Sigma_1\cup\Sigma_2...\cup\Sigma_{8},
\end{align}
which is the boundary of $\Omega_k$ respectively. In addition, let
\begin{align}
	&\Omega=\Omega_1\cup ... \cup\Omega_8.\\
	&\Sigma^{(2)}=\cup_{n\in\mathcal{N}\setminus\Lambda}\left( \partial\mathbb{D}(\bar{\zeta}_n,\varrho)\cup\partial\mathbb{D}(\zeta_n,\varrho)\right)  ,
\end{align}
which are shown in Figure \ref{figR2}.
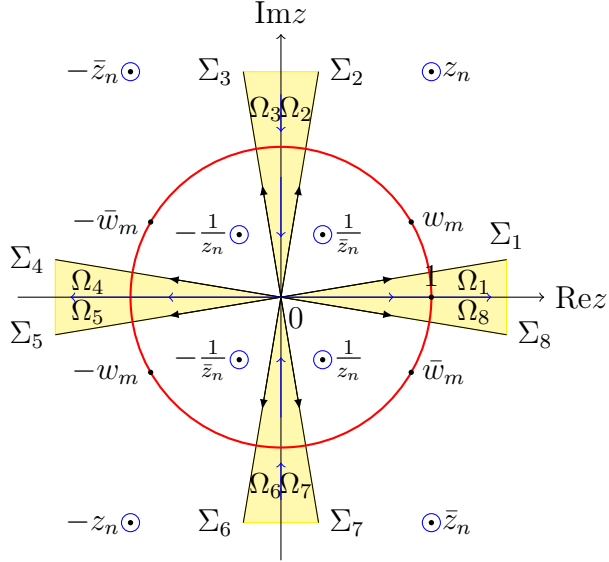
\begin{figure}[htp]
	\centering
		\begin{tikzpicture}[node distance=2cm]
		\draw[yellow, fill=yellow!40] (0,0)--(0.5,3)--(-0.5,3)--(0,0)--(0.5,-3)--(-0.5,-3)--(0,0);
		\draw[yellow, fill=yellow!40] (0,0)--(3,-0.5)--(3,0.5)--(0,0)--(-3,-0.5)--(-3,0.5)--(0,0);
		\draw(0,0)--(3,0.5)node[above]{$\Sigma_1$};
		\draw(0,0)--(0.5,3)node[right]{$\Sigma_2$};
		\draw(0,0)--(-0.5,3)node[left]{$\Sigma_3$};
		\draw(0,0)--(-3,0.5)node[left]{$\Sigma_4$};
		\draw(0,0)--(-3,-0.5)node[left]{$\Sigma_5$};
		\draw(0,0)--(-0.5,-3)node[left]{$\Sigma_6$};
		\draw(0,0)--(0.5,-3)node[right]{$\Sigma_7$};
		\draw(0,0)--(3,-0.5)node[right]{$\Sigma_8$};
		\draw[->](-3.5,0)--(3.5,0)node[right]{ Re$z$};
		\draw[->](0,-3.5)--(0,3.5)node[above]{ Im$z$};
		\draw[-latex](0,0)--(-1.5,-0.25);
		\draw[-latex](0,0)--(-1.5,0.25);
		\draw[-latex](0,0)--(1.5,0.25);
		\draw[-latex](0,0)--(1.5,-0.25);
		\draw[-latex](0,0)--(0.25,-1.5);
		\draw[-latex](0,0)--(0.25,1.5);
		\draw[-latex](0,0)--(-0.25,1.5);
		\draw[-latex](0,0)--(-0.25,-1.5);
		\coordinate (I) at (0.2,0);
		\coordinate (C) at (-0.2,2.2);
		\fill (C) circle (0pt) node[above] {\small $\Omega_3$};
		\coordinate (E) at (0.2,2.2);
		\fill (E) circle (0pt) node[above] {\small $\Omega_2$};
		\coordinate (D) at (2.2,0.2);
		\fill (D) circle (0pt) node[right] {\small$\Omega_1$};
		\coordinate (F) at (-0.2,-2.2);
		\fill (F) circle (0pt) node[below] {\small$\Omega_6$};
		\coordinate (J) at (-2.2,-0.2);
		\fill (J) circle (0pt) node[left] {\small$\Omega_5$};
		\coordinate (k) at (-2.2,0.2);
		\fill (k) circle (0pt) node[left] {\small$\Omega_4$};
		\coordinate (J) at (0.2,-2.2);
		\fill (J) circle (0pt) node[below] {\small$\Omega_7$};
		\coordinate (k) at (2.2,-0.2);
		\fill (k) circle (0pt) node[right] {\small$\Omega_8$};
		\fill (I) circle (0pt) node[below] {$0$};
		\draw[red,  thick] (2,0) arc (0:360:2);
		\draw[blue] (2,3) circle (0.12);
		\draw[blue][->](0,0)--(-1.5,0);
		\draw[blue][->](-1.5,0)--(-2.8,0);
		\draw[blue][->](0,0)--(1.5,0);
		\draw[blue][->](1.5,0)--(2.8,0);
		\draw[blue][->](0,2.7)--(0,2.2);
		\draw[blue][->](0,1.6)--(0,0.8);
		\draw[blue][->](0,-2.7)--(0,-2.2);
		\draw[blue][->](0,-1.6)--(0,-0.8);
		\coordinate (A) at (2,3);
		\coordinate (B) at (2,-3);
		\coordinate (C) at (-0.5546996232,0.8320505887);
		\coordinate (D) at (-0.5546996232,-0.8320505887);
		\coordinate (E) at (0.5546996232,0.8320505887);
		\coordinate (F) at (0.5546996232,-0.8320505887);
		\coordinate (G) at (-2,3);
		\coordinate (H) at (-2,-3);
		\coordinate (I) at (2,0);
		\draw[blue] (2,-3) circle (0.12);
		\draw[blue] (-0.55469962326,0.8320505887) circle (0.12);
		\draw[blue] (0.5546996232,0.8320505887) circle (0.12);
		\draw[blue] (-0.5546996232,-0.8320505887) circle (0.12);
		\draw[blue] (0.5546996232,-0.8320505887) circle (0.12);
		\draw[blue] (-2,3) circle (0.12);
		\draw[blue] (-2,-3) circle (0.12);
		\coordinate (J) at (1.7320508075688774,1);
		\coordinate (K) at (1.7320508075688774,-1);
		\coordinate (L) at (-1.7320508075688774,1);
		\coordinate (M) at (-1.7320508075688774,-1);
		\fill (A) circle (1pt) node[right] {$z_n$};
		\fill (B) circle (1pt) node[right] {$\bar{z}_n$};
		\fill (C) circle (1pt) node[left] {$-\frac{1}{z_n}$};
		\fill (D) circle (1pt) node[left] {$-\frac{1}{\bar{z}_n}$};
		\fill (E) circle (1pt) node[right] {$\frac{1}{\bar{z}_n}$};
		\fill (F) circle (1pt) node[right] {$\frac{1}{z_n}$};
		\fill (G) circle (1pt) node[left] {$-\bar{z}_n$};
		\fill (H) circle (1pt) node[left] {$-z_n$};
		\fill (I) circle (1pt) node[above] {$1$};
		\fill (J) circle (1pt) node[right] {$w_m$};
		\fill (K) circle (1pt) node[right] {$\bar{w}_m$};
		\fill (L) circle (1pt) node[left] {$-\bar{w}_m$};
		\fill (M) circle (1pt) node[left] {$-w_m$};
		\end{tikzpicture}
	\caption{The yellow region is $\Omega$. The blue circle constitute $\Sigma^{(2)}$ together. }
	\label{figR2}
\end{figure}

\begin{lemma}\label{Imtheta}
	Let $\xi=\frac{x}{t}\in(-3,-1)$, and $F(r)=r^2+\frac{1}{r^2}$ is a  real-valued function. Then for $z=re^{i\phi}$, the imaginary part of phase function (\ref{Reitheta}) satisfies
	\begin{align}
		&\text{Im }\theta(z)\leq \frac{1}{16}|\sin 2\phi|(|\xi+2|-1)F(r)^2,\hspace{0.5cm} \text{as }z\in\Omega_1, \Omega_3, \Omega_5, \Omega_7;\\
		&\text{Im }\theta(z)\geq  \frac{1}{16}|\sin 2\phi|(1-|\xi+2|)F(r)^2,\hspace{0.5cm} \text{as }z\in\Omega_2, \Omega_4, \Omega_6, \Omega_8.
	\end{align}
\end{lemma}
\begin{proof}
	We only prove the case $z\in\Omega_1,$ and the other regions are similarly.  From (\ref{Reitheta}) we have
	\begin{align}
		\text{Im }\theta(z)&=\frac{1}{2}\text{Im}z\text{Re}z\left[ \left( \xi+2\right)\left(1+|z|^{-4} \right) -\left( \text{Re}^2z-\text{Im}^2z\right)\left( 1+|z|^{-8}\right)   \right] \nonumber\\
		&=\frac{1}{4}r^2\sin 2\phi\left[\left( \xi+2\right)\left(1+r^{-4} \right) -r^2\cos  2\phi\left( 1+r^{-8}\right)   \right]  \nonumber\\
		&=\frac{1}{4}\sin 2\phi\left[\left( \xi+2\right)F(r) -\cos  2\phi\left( F(r)^2-2\right)   \right].\label{1}
	\end{align}
	 $F(r)\geq2$ leads to  $2\leq \frac{F(r)^2}{2}$. For $z\in\Omega_1$,  $\frac{2|\xi+2|}{|\xi+2|+1}<\cos 2\varphi<\cos 2\phi$, then we have
	\begin{equation}
		\frac{|\xi+2|}{\cos2 \phi}F(r)\leq \frac{|\xi+2|+1}{4}F(r)^2.
	\end{equation}
   Substitute above inequality into (\ref{1}) we obtain the consequence immediately.
\end{proof}
Introduce a small enough constant $1>\epsilon_0>0$ with $(1-\epsilon_0)\cos\varphi>\frac{1}{2}$. Let $X_1\in C_0^\infty\left(\mathbb{R},[0,1] \right) $, which is support in $(1-\epsilon_0,1+\epsilon_0)$. And $X_0$ has  support in $(-\epsilon_0,\epsilon_0)$ with $X_0(z)=X_1(1+z)$. In addition, we denote following functions for brief:
\begin{align}
	&p_1(z)=p_5(z)=\rho(z),\hspace{0.5cm}p_2(z)=p_6(z)=\dfrac{\tilde{\rho}(z) }{1-\rho(z)\tilde{\rho}(z)},\\
	&p_3(z)=p_7(z)=\dfrac{\rho(z) }{1-\rho(z)\tilde{\rho}(z)},\hspace{0.5cm}p_4(z)=p_{8}(z)=\tilde{\rho}(z).
\end{align}
Then the next step is to construct a matrix function $R^{(2)}$. We need to remove jump on $\mathbb{R}$ and $i\mathbb{R}$, and  have some mild control on $\bar{\partial}R^{(2)}$ sufficient to ensure that the $\bar{\partial}$-contribution to the long-time asymptotics of $q(x, t)$ is negligible.
So we choose $R^{(2)}(z)$ as
\begin{equation}
R^{(2)}(z)=\left\{\begin{array}{lll}
\left(\begin{array}{cc}
1 & R_j(z)e^{2it\theta}\\
0 & 1
\end{array}\right), & z\in \Omega_j,j=2,4,6,8;\\
\\
\left(\begin{array}{cc}
1 & 0\\
R_j(z)e^{-2it\theta} & 1
\end{array}\right),  &z\in \Omega_j,j=1,3,5,7;\\
\\
I,  &elsewhere;\\
\end{array}\right.\label{R(2)}
\end{equation}
where  the functions $R_j$, $j=1,2,..,8$, is defined in following Proposition.
\begin{Proposition}\label{proR}
	 $R_j$: $\bar{\Omega}_j\to C$, $j=1,2,..,8$ have boundary values as follow:
	\begin{align}
	&R_1(z)=\Bigg\{\begin{array}{ll}
	-\rho(z)T(z)^{2} & z\in \mathbb{R}^+,\\
	0  &z\in \Sigma_1,\\
	\end{array} ,\hspace{0.6cm}
	R_2(z)=\Bigg\{\begin{array}{ll}
	0  &z\in \Sigma_2,\\
	\dfrac{\tilde{\rho}(z) T_+(z)^2}{1-\rho(z)\tilde{\rho}(z)} &z\in  i\mathbb{R}^+,\\
	\end{array} \\
	&R_3(z)=\Bigg\{\begin{array}{ll}
	\dfrac{\rho(z) T_-(z)^2}{1-\rho(z)\tilde{\rho}(z)} &z\in i\mathbb{R}^+, \\
	0 &z\in \Sigma_3,\\
	\end{array} ,
	R_4(z)=\Bigg\{\begin{array}{ll}
	0  &z\in \Sigma_4,\\
	-\tilde{\rho}(z)T(z)^{-2} &z\in  \mathbb{R}^-,\\
	\end{array} \\
	&R_5(z)=\Bigg\{\begin{array}{ll}
	-\rho(z)T(z)^{2} &z\in  \mathbb{R}^-,\\
	0  &z\in \Sigma_5,
	\end{array} ,\hspace{0.5cm}
	R_6(z)=\Bigg\{\begin{array}{ll}
	0  &z\in \Sigma_6,\\
	\dfrac{\tilde{\rho}(z) T_+(z)^2}{1-\rho(z)\tilde{\rho}(z)} &z\in  i\mathbb{R}^-,\\
	\end{array} \\
	&R_7(z)=\Bigg\{\begin{array}{ll}
	\dfrac{\rho(z) T_-(z)^2}{1-\rho(z)\tilde{\rho}(z)} &z\in i\mathbb{R}^-, \\
	0  &z\in \Sigma_7,
	\end{array} ,\hspace{0.6cm}
	R_8(z)=\Bigg\{\begin{array}{ll}
	0  &z\in \Sigma_8,\\
	-\tilde{\rho}(z)T(z)^{-2} &z\in  \mathbb{R}^+.\\
	\end{array}
	\end{align}	
	 $R_j$  have following property:
	for $j=1,5,4,8,$
	\begin{align}
	&|\bar{\partial}R_j(z)|\lesssim|p_j'(|z|)|+|z|^{-1/2}, \text{for all $z\in \Omega_j$;}\label{dbarRj}
	\end{align}
	and for $j=2,3,6,7,$
	\begin{align}
	&|\bar{\partial}R_j(z)|\lesssim |z\mp i|,\text{for all $z\in \Omega_j$ in a small fixed neighborhood  of $\pm i$}, \label{Ri}\\
	&| \bar{\partial}R_j(z)|\lesssim|p_j'(i|z|)|+|z|^{-1/2}+|\bar{\partial}X_1(|z|)|,\text{for all $z\in \Omega_j$}.\label{dbarRk}
	\end{align}
	And
	\begin{equation}
	\bar{\partial}R_j(z)=0,\hspace{0.5cm}\text{if } z\in elsewhere.
	\end{equation}
\end{Proposition}

\begin{proof}
	Case I: $z\in \bar{\Omega}_j$, $j=1,5,4,8$.\\
	 Take $R_1(z)$ as an example with extensions
	\begin{equation}
		R_1(z)=p_1(|z|)T^2(z)\cos(k_0 \arg z),\hspace{0.5cm}k_0=\frac{2\pi}{\varphi}.
	\end{equation}
	The other cases are easily inferred. $p_1(|z|)=\rho(|z|)$ is bounded. Denote $z=re^{i\phi}$, then we have $\bar{\partial}=\frac{e^{i\phi}}{2}\left(\partial_r+\frac{i}{r} \partial_\phi\right) $. So
	\begin{align}
	\bar{\partial}R_1(z)=\frac{e^{i\phi}}{2}T^2(z)\left(p_1'(r)\cos(k_0\phi)-\frac{i}{r}p_1(r)k_0\sin(k_0\phi) \right) .
	\end{align}
	To bound second term we use Cauchy-Schwarz inequality and obtain
	\begin{equation}
	|p_1(r)|= |\rho(r)|= |\rho(r)-\rho(0)|=|\int_{0}^r\rho'(s)ds|\leq \parallel \rho'(s)\parallel_{L^2} r^{1/2}.
	\end{equation}
	And note that $T(z)$ is a bounded function in $\bar{\Omega}_1$. Then the boundedness of (\ref{dbarRj})  follows immediately.\\
	Case II: $z\in \bar{\Omega}_j$, $j=2,3,6,7.$\\	
	The details of the proof are only given  for $R_2$. Unlike the vanishing boundary condition case in \cite{fNLS}, the determinant of $M(z)$ is $1+z^{-2}$. So to bound the $\bar{\partial}$-derivative construct by $R^{(2)}$ in following section,   the property of $R^{(2)}$ at $\pm i$ needs to be control. For this purpose, we make small adjustments to the extensions of $R_2$ as
	\begin{equation}
		R_2(z)=R_{21}(z)+R_{22}(z),
	\end{equation}
	with a constant $\delta_0$ stratifying $\varphi>\delta_0\epsilon_0$ and
	\begin{align}
		R_{21}(z)&=[1-X_1(|z|)]p_2(i|z|)T^{-2}(z)\cos[k_0(\frac{\pi}{2}-\arg z)],\\
		R_{22}(z)&=f(|z|)g(z)\cos[k_0(\frac{\pi}{2}-\arg z)]\nonumber\\
		&-\frac{i|z|}{k_0}X_0(\frac{\arg z}{\delta_0})f'(|z|)g(z)\sin[k_0(\frac{\pi}{2}-\arg z)].
	\end{align}
	Among above function,
	\begin{align}
		f(z)=X_1(z)\frac{\bar{b}(z)}{a(z)},\hspace{0.5cm}g(z)=\left( \frac{a(z)}{T(z)}\right) ^2.
	\end{align}
	Then $f(z)\in W{2,\infty}$. Obviously, $R_{21}(z)\equiv0$ with $|z|$ in the support of $X_1$ and $R_{22}(z)\equiv0$  out the support of $X_1$. Note that
	\begin{align}
		|p_2(z)|=|\dfrac{\tilde{\rho}(z) }{1-\rho(z)\tilde{\rho}(z)}|=|\dfrac{\tilde{\rho}(z) }{1-|\rho(z)|^2}|\lesssim |\rho(z)|,\hspace{0.5cm}\text{for $z$ out of supp}(X_1).
	\end{align}
	Similarly in case I, $R_{21}(z)$ can be bounded as
	\begin{equation}
		|\bar{\partial}R_{21}(z)|\lesssim (1-X_1(|z|))\left( |p_2'(i|z|)|+|z|^{-1/2}\right)+|\bar{\partial}X_1(|z|)| .
	\end{equation}
	As for $R_{22}(z)$, $z=re^{i\phi}$,
	\begin{align}
		\bar{\partial}R_{22}(z)=&\frac{e^{i\phi}}{2}g(z) \cos[k_0(\frac{\pi}{2}-\varphi)]f'(ir)\left( 1-X_0(\frac{\varphi}{\delta_0})\right)  \nonumber\\
		&+ \sin[k_0(\frac{\pi}{2}-\varphi)]\left[\frac{ik_0}{r}f(ir)+\frac{1}{\delta_0 k_0}X_0'(\frac{\varphi}{\delta_0})f'(ir) \right] \nonumber\\
		&-\frac{i}{k_0}\sin[k_0(\frac{\pi}{2}-\varphi)]X_0(\frac{\arg z}{\delta_0})(rf'(ir))'.
	\end{align}
	So $|\bar{\partial}R_{22}(z)|$ is bounded, and we can write $|\bar{\partial}R_{22}(z)|\lesssim X_1(z)|z|^{-1/2}$. So (\ref{Ri}) is obtained. In addition, for $z\sim i$,
	\begin{equation}
		|\bar{\partial}R_{22}(z)|\lesssim |\sin[k_0(\frac{\pi}{2}-\varphi)]|+|1-X_0(\frac{\varphi}{\delta_0})|=\mathcal{O}(\varphi),
	\end{equation}
	from which (\ref{Ri})  follows immediately.
\end{proof}
In addition, from Proposition \ref{sym}, $R^{(2)}$ achieve the symmetry:
\begin{equation}
	R^{(2)}(z)=\sigma_2\overline{R^{(2)}(\bar{z})}\sigma_2=\sigma_1\overline{R^{(2)}(-\bar{z})}\sigma_1=\sigma_3Q_-R^{(2)}(-1/z)\sigma_3Q_-.
\end{equation}
We now  use $R^{(2)}$ to define the new transformation \begin{equation}
	M^{(2)}(z)=M^{(1)}(z)R^{(2)}(z)\label{transm2},
\end{equation}
which satisfies the following mixed $\bar{\partial}$-RH problem.

\noindent \textbf{RHP2}. Find a matrix valued function  $ M^{(2)}(z;x,t)$ with following properties:

$\blacktriangleright$ Analyticity:  $M^{(2)}(z;x,t)$ is continuous in $\mathbb{C}$,  sectionally continuous first partial derivatives in
$\mathbb{C}\setminus \left( \Sigma^{(2)}\cup \left\lbrace\zeta_n,\bar{\zeta}_n \right\rbrace_{n\in\Lambda} \right) $  and meromorphic out $\bar{\Omega}$;

$\blacktriangleright$ Symmetry: $M^{(2)}(z)=\sigma_2\overline{M^{(2)}(\bar{z})}\sigma_2$=$\sigma_1\overline{M^{(2)}(-\bar{z})}\sigma_1=\frac{i}{z}M^{(2)}(-1/z)\sigma_3Q_-$;

$\blacktriangleright$ Jump condition: $M^{(2)}$ has continuous boundary values $M^{(2)}_\pm$ on $\Sigma^{(2)}$ and
\begin{equation}
	M^{(2)}_+(z;x,t)=M^{(2)}_-(z;x,t)V^{(2)}(z),\hspace{0.5cm}z \in \Sigma^{(2)},
\end{equation}
where
\begin{equation}
	V^{(2)}(z)=\left\{ \begin{array}{ll}
		\left(\begin{array}{cc}
			1 & 0\\
			-C_n(z-\zeta_n)^{-1}T^2(z)e^{-2it\theta_n} & 1
		\end{array}\right),   &\text{as } 	z\in\partial\mathbb{D}(\zeta_n,\varrho),n\in\nabla\setminus\Lambda;\\[12pt]
		\left(\begin{array}{cc}
			1 & -C_n^{-1}(z-\zeta_n)T^{-2}(z)e^{2it\theta_n}\\
			0 & 1
		\end{array}\right),   &\text{as } z\in\partial\mathbb{D}(\zeta_n,\varrho),n\in\Delta;\\
		\left(\begin{array}{cc}
			1 & \bar{C}_n(z-\bar{\zeta}_n)^{-1}T^{-2}(z)e^{2it\bar{\theta}_n}\\
			0 & 1
		\end{array}\right),   &\text{as } 	z\in\partial\mathbb{D}(\bar{\zeta}_n,\varrho),n\in\nabla\setminus\Lambda;\\
		\left(\begin{array}{cc}
			1 & 0	\\
			\bar{C}_n^{-1}(z-\bar{\zeta}_n)e^{-2it\bar{\theta}_n}T^2(z) & 1
		\end{array}\right),   &\text{as } 	z\in\partial\mathbb{D}(\bar{\zeta}_n,\varrho),n\in\Delta;\\
	\end{array}\right.;\label{jumpv2}
\end{equation}

$\blacktriangleright$ Asymptotic behaviors:
\begin{align}
	&M^{(2)}(z) = e^{i\nu_-(x,t;q)\sigma_3}+\mathcal{O}(z^{-1}),\hspace{0.5cm}z \rightarrow \infty,\\
	&M^{(2)}(z) =\frac{i}{z}e^{i\nu_-(x,t;q)\sigma_3}\sigma_3Q_-+\mathcal{O}(1),\hspace{0.5cm}z \rightarrow 0;
\end{align}

$\blacktriangleright$ $\bar{\partial}$-Derivative: For $z\in\mathbb{C}$
we have
\begin{align}
\bar{\partial}M^{(2)}=M^{(2)}\bar{\partial}R^{(2)},
\end{align}
where
\begin{equation}
\bar{\partial}R^{(2)}=\left\{\begin{array}{lll}
\left(\begin{array}{cc}
0 & \bar{\partial}R_j(z)e^{2it\theta}\\
0 & 0
\end{array}\right), & z\in \Omega_j,j=1,3,5,7,\\
\\
\left(\begin{array}{cc}
0 & 0\\
\bar{\partial}R_j(z)e^{-2it\theta} & 0
\end{array}\right),  &z\in \Omega_j,j=2,4,6,8,\\
\\
0  &elsewhere;\\
\end{array}\right. \label{DBARR2}
\end{equation}

$\blacktriangleright$ Residue conditions: $M^{(2)}$ has simple poles at each point $\zeta_n$ and $\bar{\zeta}_n$ for $n\in\Lambda$ with:
\begin{align}
	&\res_{z=\zeta_n}M^{(2)}(z)=\lim_{z\to \zeta_n}M^{(2)}(z)\left(\begin{array}{cc}
		0 & 0\\
		C_ne^{-2it\theta_n}T^2(\zeta_n) & 0
	\end{array}\right),\\
	&\res_{z=\bar{\zeta}_n}M^{(2)}(z)=\lim_{z\to \bar{\zeta}_n}M^{(2)}(z)\left(\begin{array}{cc}
		0 & -\bar{C}_nT^{-2}(\bar{\zeta}_n)e^{2it\bar{\theta}_n}\\
		0 & 0
	\end{array}\right).
\end{align}

\section{ Decomposition of the mixed $\bar{\partial}$-RH problem }\label{sec5}
\quad To solve RHP2,  we decompose it into a model   RH  problem  for $M^{(r)}(z)$  with $\bar\partial R^{(2)}\equiv0$   and a pure $\bar{\partial}$-Problem with nonzero $\bar{\partial}$-derivatives.
For the first step, we establish  a   RH problem  for the  $M^{(r)}(z)$   as follows.

\noindent\textbf{RHP3}. Find a matrix-valued function  $ M^{(r)}(z)$ with following properties:

$\blacktriangleright$ Analyticity: $M^{(r)}(z)$ is  meromorphic  in $\mathbb{C}\setminus \Sigma^{(2)}$;

$\blacktriangleright$ Jump condition: $M^r$ has continuous boundary values $M^{(r)}_\pm$ on $\Sigma^{(2)}$ and
\begin{equation}
M^{(r)}_+(z)=M^{(r)}_-(z)V^{2}(z),\hspace{0.5cm}z \in \Sigma^{(2)};\label{jump5}
\end{equation}

$\blacktriangleright$ Symmetry: $M^{(r)}(z)=\sigma_2\overline{M^{(r)}(\bar{z})}\sigma_2$=$\sigma_1\overline{M^{(r)}(-\bar{z})}\sigma_1=\frac{i}{z}M^{(r)}(-1/z)\sigma_3Q_-$;

$\blacktriangleright$ $\bar{\partial}$-Derivative:  $\bar{\partial}R^{(2)}=0$, for $ z\in \mathbb{C}$;

$\blacktriangleright$ Asymptotic behaviors:
\begin{align}
	&M^{(r)}(z) = e^{i\nu_-(x,t;q)\sigma_3}+\mathcal{O}(z^{-1}),\hspace{0.5cm}z \rightarrow \infty,\\
	&M^{(r)}(z) =\frac{i}{z}e^{i\nu_-(x,t;q)\sigma_3}\sigma_3Q_-+\mathcal{O}(1),\hspace{0.5cm}z \rightarrow 0;\label{asymbehv6}
\end{align}

$\blacktriangleright$ Residue conditions: $M^{(r)}$ has simple poles at each point $\zeta_n$ and $\bar{\zeta}_n$ for $n\in\Lambda$ with:
\begin{align}
	&\res_{z=\zeta_n}M^{(r)}(z)=\lim_{z\to \zeta_n}M^{(r)}(z)\left(\begin{array}{cc}
		0 & 0\\
		C_ne^{-2it\theta_n}T^2(\zeta_n) & 0
	\end{array}\right),\\
	&\res_{z=\bar{\zeta}_n}M^{(r)}(z)=\lim_{z\to \bar{\zeta}_n}M^{(r)}(z)\left(\begin{array}{cc}
		0 & -\bar{C}_nT^{-2}(\bar{\zeta}_n)e^{2it\bar{\theta}_n}\\
		0 & 0
	\end{array}\right).\label{resMr}
\end{align}

The unique existence  and asymptotic  of  $M^{(r)}(z)$  will shown in   section \ref{sec6}.

We now use $M^{(r)}(z)$ to construct  a new matrix function
\begin{equation}
M^{(3)}(z)=M^{(2)}(z)M^{(r)}(z)^{-1},\label{transm3}
\end{equation}
which   removes   analytical component  $M^{(r)}(z)$    to get  a  pure $\bar{\partial}$-problem.

\noindent\textbf{$\bar{\partial}$-problem4}. Find a matrix-valued function  $M^{(3)}(z)$  with following properties:

$\blacktriangleright$ Analyticity: $M^{(3)}(z)$ is continuous   and has sectionally continuous first partial derivatives in $\mathbb{C}$.

$\blacktriangleright$ Asymptotic behavior:
\begin{align}
&M^{(3)}(z) \sim I+\mathcal{O}(z^{-1}),\hspace{0.5cm}z \rightarrow \infty;\label{asymbehv7}
\end{align}

$\blacktriangleright$ $\bar{\partial}$-Derivative: We have
$$\bar{\partial}M^{(3)}(z)=M^{(3)}(z)W,\ \ z\in \mathbb{C},$$
where
\begin{equation}
W=M^{(r)}(z)\bar{\partial}R^{(2)}(z)M^{(r)}(z)^{-1}.
\end{equation}

\begin{proof}
	By using  properties  of  the   solutions   $M^{(2)}(z)$ and $M^{(r)}(z)$  for  RHP3  and $\bar{\partial}$-problem 4,
 the analyticity is obtained   immediately. And for its Asymptotic behavior, from $M^{(r)}(z)^{-1}=(1+z^{-2})\sigma_2M^{(r)}(z)^T\sigma_2$ we have
 \begin{align}
 	\lim_{z\to 0}M^{(3)}(z)=&\lim_{z\to 0}\dfrac{(zM^{(2)}(z))\sigma_2(zM^{(r)}(z)^T)\sigma_2}{1+z^2}\nonumber\\
 	=&ie^{i\nu_-(x,t;q)\sigma_3}\sigma_3Q_-\sigma_2(ie^{i\nu_-(x,t;q)\sigma_3}\sigma_3Q_-)^T\sigma_2=I.
 \end{align}
Since $M^{(2)}(z)$ and $M^{(r)}(z)$ achieve same jump matrix, we have
	\begin{align*}
	M_-^{(3)}(z)^{-1}M_+^{(3)}(z)&=M_-^{(2)}(z)^{-1}M_-^{(r)}(z)M_+^{(r)}(z)^{-1}M_+^{(2)}(z)\\
	&=M_-^{(2)}(z)^{-1}V^{(2)}(z)^{-1}M_+^{(2)}(z)=I,
	\end{align*}
	which implies $ M^{(3)}(z)$ has no jumps and is everywhere continuous.  We also can show  that $   M^{(3)}(z)$ has no pole.  For
 $\lambda \in \left\lbrace \zeta_n,\bar{\zeta}_n \right\rbrace_{n\in\Lambda} $,  let $\mathcal{N}_\lambda$ denote the  nilpotent matrix which appears in the left side of the
corresponding residue condition of RHP4  and  RHP5,
 we have the Laurent expansions in $z-\lambda$
	\begin{align}
&M^{(2)}(z)=a(\lambda) \left[ \dfrac{\mathcal{N}_\lambda}{z-\lambda}+I\right] +\mathcal{O}(z-\lambda),\nonumber\\
&	M^{(r)}(z)=A(\lambda) \left[ \dfrac{\mathcal{N}_\lambda}{z-\lambda}+I\right] +\mathcal{O}(z-\lambda),\nonumber
\end{align}
	where $a(\lambda)$ and $A(\lambda)$ are the constant  matrix in their respective expansions.
Then
	\begin{align}
	M^3(z)&=\left\lbrace a(\lambda) \left[ \dfrac{\mathcal{N}_\lambda}{z-\lambda}+I\right]\right\rbrace \left\lbrace\left[ \dfrac{-\mathcal{N}_\lambda}{z-\lambda}+I\right]\sigma_2A(\lambda)^T\sigma_2\right\rbrace + \mathcal{O}(z-\lambda)\nonumber\\
	&=\mathcal{O}(1),
	\end{align}
	which  implies that  $M^{(3)}(z)$ has removable singularities at $\lambda$.
 And the $\bar{\partial}$-derivative of  $M^{(3)}(z)$ come  from    $M^{(3)}(z)$  due to
   analyticity of $M^{(r)}(z)$. In addition, unlike the zero boundary case, we must check its property at $\pm i$. The symmetries of $M^{(2)}(z)$ and $M^{(r)}(z)$ imply that
 \begin{align}
 	M^{(2)}(z)=&\left(\begin{array}{cc}
 		\gamma & \pm q_-\gamma\\
 		\pm \bar{q}_-\bar{\gamma}& \bar{\gamma}
 	\end{array}\right)+\mathcal{O}(z\mp i),\\
 M^{(r)}(z)=&\dfrac{\pm i}{2(z\mp i)}\left(\begin{array}{cc}
 	\bar{\iota}& \mp q_-\iota\\
 	\mp \bar{q}_-\bar{\iota}& \iota
 \end{array}\right)+\mathcal{O}(1),
 \end{align}
for two constants $\gamma$ and $\iota$. Then the singular part of $M^{(3)}(z)$ vanishes at $z=\pm i$ by simple calculation  immediately.
\end{proof}
The unique existence  and asymptotic  of  $M^{(3)}(z)$  will shown in   section \ref{sec7}.

\section{ Asymptotic of $\mathcal{N}(\Lambda)$-soliton solutions } \label{sec6}
\quad In this section, we build a reflectionless  RH  problem and  show that  its solution can  approximated  with  $M^{(r)}$.

First we  show the existence and uniqueness of solution of the above RHP3 which is  related   with original RH problem 0.
\begin{Proposition}
	  The solution $M^{(r)}(z)$ of the RH problem 3 with scattering data $\left\lbrace  r(z),\left\lbrace \zeta_n,C_n\right\rbrace_{n\in\Lambda}\right\rbrace$  exists  and is unique.
By an explicit transformation, $M^{(r)}(z)$ is equivalent  to a reflectionless solution of the original RHP0 with modified scattering data $\left\lbrace  0,\left\lbrace \zeta_n,\mathring{c}_n\right\rbrace_{n\in\Lambda}\right\rbrace$, where
	\begin{equation}
		\mathring{c}_n(x,t)=C_n\exp\left\lbrace-\frac{1}{i\pi}\int_{\mathbb{R}}\log(1-|\rho(s)|^2)\left(\frac{1}{s-\zeta_n}-\frac{1}{2s} \right)  \right\rbrace .
	\end{equation}
\end{Proposition}
\begin{proof}
	To transform $M^{(r)}(z)$ to the soliton-solution  of RHP0, the jumps and poles need to be restored. We reverses the triangularity effected in (\ref{transm1}) and (\ref{transm2}):
	\begin{equation}
		N(z)=\left(\prod_{n\in \Delta}\zeta_n \right)^{-\sigma_3} M^{(r)}(z)T^{-\hat{\sigma}_3}G^{-1}(z)\left( \prod_{n\in \Delta}\dfrac{z-\zeta_n}{\bar{\zeta}_n^{-1}z-1}\right) ^{-\sigma_3},
	\end{equation}
with $G(z)$ defined in (\ref{funcG}). First we verify $N(z)$ satisfying RHP0. This transformation to $N(z)$ preserves the normalization conditions at the origin and infinity obviously. And comparing with (\ref{transm1}), this transformation  restore the jump on   $\mathbb{D}(\bar{\zeta}_n,\varrho)$ and $\mathbb{D}(\zeta_n,\varrho)$ to residue for $n\notin\Lambda$.  As for $n\in\Lambda$, take $\zeta_n$ as an example. Substitute (\ref{resMr}) into the transformation:
\begin{align}
	\res_{z=\zeta_n}N(z)=&\left(\prod_{n\in \Delta}\zeta_n \right)^{-\sigma_3}\res_{z=\zeta_n}M^{(r)}(z)T^{-\hat{\sigma}_3}G(z)^{-1}\left( \prod_{n\in \Delta}\dfrac{z-\zeta_n}{\bar{\zeta}_n^{-1}z-1}\right) ^{-\sigma_3}\nonumber\\
	=&\lim_{z\to \zeta_n}-\left(\prod_{n\in \Delta}\zeta_n \right)^{-\sigma_3}M^{(r)}(z)\left(\begin{array}{cc}
		0 & 0\\
		C_ne^{-2it\theta_n}T^2(\zeta_n) & 0
	\end{array}\right)\left( \prod_{n\in \Delta}\dfrac{z-\zeta_n}{\bar{\zeta}_n^{-1}z-1}\right) ^{-\sigma_3}\nonumber\\
		=&\lim_{z\to \zeta_n}N(z)\left(\begin{array}{cc}
			0 & 0\\
			\mathring{c}_ne^{-2it\theta_n} & 0
		\end{array}\right).
\end{align}
Its analyticity and symmetry follow from the Proposition of $M^{(r)}(z)$, $T(z)$ and $G(z)$ immediately. So $N(z)$ is solution of RHP0 with absence of reflection, whose unique exact solution  exists and can be obtained as described similarly in \cite{SandRNLS}. So $M^{(r)}(z)$ unique exists.
\end{proof}

Although $M^{(r)}(z)$  admits  uniqueness and  existence,  we can't  give its explicit expression. The jump matrix  is uniformly near identity and doesn't
contribute to the asymptotic behavior of the solution.
\begin{lemma}\label{lemmav2}
	The jump matrix $ V^{(2)}(z)$ in (\ref{jumpv2}) satisfies
	\begin{align}
	&\parallel V^{(2)}-I\parallel_{L^\infty(\Sigma^{(2)})}=\mathcal{\mathcal{O}}(e^{- 2\rho_0t} ),\label{7.1}
\end{align}
where    $\rho_0$ is defined  by (\ref{rho0}).
\end{lemma}
\begin{proof}
	Take $z\in\partial\mathbb{D}(\zeta_n,\varrho),$ $n\in\nabla\setminus\Lambda$ as an example.
	\begin{align}
		\parallel V^{(2)}-I\parallel_{L^\infty(\partial\mathbb{D}(\zeta_n,\varrho))}&=|C_n(z-\zeta_n)^{-1}T^2(z)e^{-2it\theta_n}|\nonumber\\
		&\lesssim \varrho^{-1}e^{-\text{Re}(2it\theta_n)}\lesssim e^{2t\text{Im}(\theta_n)}\leq e^{-2\rho_0t}.
	\end{align}
The last step follows from that for $n\in\nabla\setminus\Lambda$, $ \text{Im}\theta_n<0$.
\end{proof}
\begin{corollary}\label{v2p}
	For $1\leq p\leq +\infty$, the jump matrix $V^{(2)}(z)$ satisfies
	\begin{equation}
		\parallel V^{(2)}-I\parallel_{L^p(\Sigma^{(2)})}\leq K_pe^{- 2\rho_0t} ,
	\end{equation}
for some constant $K_p\geq 0$ depending on $p$.
\end{corollary}
This  estimation of $V^{(2)}$ inspires us to consider to completely ignore the jump condition on $M^{(r)}(z)$, because there is only exponentially small error (in t). We decompose $M^{(r)}(z)$ as
\begin{equation}
	M^{(r)}(z)=E(z)M^{(r)}_\Lambda(z),\label{transMr}
\end{equation}
where $E(z)$ is a error function, which is a solution of a small-norm RH problem and we discuss it in Section \ref{sec8}. $M^{(r)}_\Lambda(z)$ solves RHP3 with $V^{(2)}\equiv0$.

\noindent\textbf{RHP5.}  Find a matrix-valued function  $ M^{(r)}_\Lambda(z;x,t)$ with following properties:

$\blacktriangleright$ Analyticity: $M^{(r)}_\Lambda(z;x,t)$ is analytical  in $\mathbb{C}\setminus \left\lbrace\zeta_n,\bar{\zeta}_n \right\rbrace_{n\in\Lambda} $;

$\blacktriangleright$ Symmetry: $M^{(r)}_\Lambda(z)=\sigma_2\overline{M^{(r)}_\Lambda(\bar{z})}\sigma_2$=$\sigma_1\overline{M^{(r)}_\Lambda(-\bar{z})}\sigma_1=\frac{i}{z}M^{(r)}_\Lambda(-1/z)\sigma_3Q_-$;

$\blacktriangleright$ Asymptotic behaviors:
\begin{align}
	&M^{(r)}_\Lambda(z;x,t) = e^{i\nu_-(x,t;q)\sigma_3}+\mathcal{O}(z^{-1}),\hspace{0.5cm}z \rightarrow \infty,\\
	&M^{(r)}_\Lambda(z;x,t) =\frac{i}{z}e^{i\nu_-(x,t;q)\sigma_3}\sigma_3Q_-+\mathcal{O}(1),\hspace{0.5cm}z \rightarrow 0;\label{asymbehv8}
\end{align}

$\blacktriangleright$ Residue conditions: $M^{(r)}_\Lambda$ has simple poles at each point $\zeta_n$ and $\bar{\zeta}_n$ for $n\in\Lambda$ with:
\begin{align}
	&\res_{z=\zeta_n}M^{(r)}_\Lambda(z)=\lim_{z\to \zeta_n}M^{(r)}_\Lambda(z)\left(\begin{array}{cc}
		0 & 0\\
		C_ne^{-2it\theta_n}T^2(\zeta_n) & 0
	\end{array}\right),\\
	&\res_{z=\bar{\zeta}_n}M^{(r)}_\Lambda(z)=\lim_{z\to \bar{\zeta}_n}M^{(r)}_\Lambda(z)\left(\begin{array}{cc}
		0 & -\bar{C}_nT^{-2}(\bar{\zeta}_n)e^{2it\bar{\theta}_n}\\
		0 & 0
	\end{array}\right).\label{resMrsol}
\end{align}

\begin{Proposition}\label{unim}	 The RHP5   exists an  unique solution.  Moreover, $M^{(r)}_\Lambda(z)$ is equivalent  to a reflectionless solution of the original RHP0 with modified scattering data $\left\lbrace  0,\left\lbrace \zeta_n,\mathring{c}_n\right\rbrace_{n\in\Lambda}\right\rbrace$ as follows:\\
	\textbf{Case I}: if $\Lambda=\varnothing$, then
	\begin{equation}
		M^{(r)}_\Lambda(z)=e^{i\nu_-(x,t;q^r_\Lambda)\sigma_3}+\frac{i}{z}e^{i\nu_-(x,t;q^r_\Lambda)\sigma_3}\sigma_3Q_-;\label{msol1}
	\end{equation}
	\textbf{Case I}: if $\Lambda\neq\varnothing$ with $\Lambda_1=\left\lbrace z_{j_k}\right\rbrace_{k=1}^{n_1} $ and $\Lambda_2=\left\lbrace w_{i_s}\right\rbrace_{s=1}^{n_2}$, then
	\begin{align}
		M^{(r)}_\Lambda(z)&=e^{i\nu_-(x,t;q^r_\Lambda)\sigma_3}+\frac{i}{z}e^{i\nu_-(x,t;q^r_\Lambda)\sigma_3}\sigma_3Q_-\nonumber\\
		&+\sum_{s=1}^{n_2}\left[\left(\begin{array}{cc}
			\frac{\alpha_s}{z-w_{i_s}} & \frac{\overline{\kappa_s}}{z-\bar{w}_{i_s}}\\
			\frac{\kappa_s}{z-w_{i_s}} & \frac{\overline{\alpha_s}}{z-\bar{w}_{i_s}}
		\end{array}\right)+\left(\begin{array}{cc}
		-\frac{\alpha_s}{z+w_{i_s}} & \frac{\overline{\kappa_s}}{z+\bar{w}_{i_s}}\\
		\frac{\kappa_s}{z+w_{i_s}} & -\frac{\overline{\alpha_s}}{z+\bar{w}_{i_s}}
	\end{array}\right) \right]\nonumber\\
&+\sum_{k=1}^{n_1}\left[\left(\begin{array}{cc}
	\frac{\beta_k}{z-z_{j_k}} & \frac{\overline{\varsigma_k}}{z-\bar{z}_{j_k}}\\
	\frac{\varsigma_k}{z-z_{j_k}} & \frac{\overline{\beta_k}}{z-\bar{z}_{j_k}}
\end{array}\right)+\left(\begin{array}{cc}
	-\frac{\beta_k}{z+z_{j_k}} & \frac{\overline{\varsigma_k}}{z+\bar{z}_{j_k}}\\
	\frac{\varsigma_k}{z+z_{j_k}} & -\frac{\overline{\beta_k}}{z+\bar{z}_{j_k}}
\end{array}\right) \right]\nonumber\\
&+\sum_{k=1}^{n_1}i\left[\left(\begin{array}{cc}
	\frac{-\overline{q_-\beta_k}}{\bar{z}_{j_k}z-1} & \frac{-q_-\varsigma_k}{z_{j_k}z-1}\\
	 \frac{-\overline{q_-\varsigma_k}}{\bar{z}_{j_k}z-1} & \frac{q_-\beta_k}{z_{j_k}z-1}
\end{array}\right)+\left(\begin{array}{cc}
	\frac{\overline{q_-\beta_k}}{\bar{z}_{j_k}z+1} & \frac{-q_-\varsigma_k}{z_{j_k}z+1}\\
	\frac{-\overline{q_-\varsigma_k}}{\bar{z}_{j_k}z+1} & \frac{-q_-\beta_k}{z_{j_k}z+1}
\end{array}\right)\right]  ,\label{msol2}
	\end{align}
where $\beta_k=\beta_k(x,t)$, $\varsigma_k=\varsigma_k(x,t)$, $\alpha_s=\alpha_s(x,t)$ and $\kappa_s=\kappa_s(x,t)$   with linearly dependant equations:
\begin{align}
	 c_{j_k}^{-1}T(z_{j_k})^{-2}e^{-2i\theta(z_{j_k})t}\beta_k&=\frac{i}{z_{j_k}}e^{i\nu_-(x,t;q^r_\Lambda)}q_-+\sum_{h=1}^{n_2}\left(\frac{\overline{\kappa_h}}{z_{j_k}-\bar{w}_{i_h}}+\frac{\overline{\kappa_h}}{z_{j_k}+\bar{w}_{i_h}} \right) \nonumber\\
	&+\sum_{l=1}^{n_1}\left( \frac{\overline{\varsigma_l}}{z_{j_k}-\bar{z}_{j_l}} +\frac{\overline{\varsigma_l}}{z_{j_k}+\bar{z}_{j_l}}-\frac{iq_-\varsigma_l}{z_{j_l}z_{j_k}-1}-\frac{iq_-\varsigma_l}{z_{j_l}z_{j_k}+1}\right) , \\	
	 c_{j_k}^{-1}T(z_{j_k})^{-2}e^{-2i\theta(z_{j_k})t}\varsigma_k&=\frac{i}{z_{j_k}}e^{-i\nu_-(x,t;q^r_\Lambda)}\bar{q}_-+\sum_{h=1}^{n_2}\left(\frac{\overline{\alpha_h}}{z_{j_k}-\bar{w}_{i_h}}-\frac{\overline{\alpha_h}}{z_{j_k}+\bar{w}_{i_h}} \right) \nonumber\\
	&+\sum_{l=1}^{n_1}\left( \frac{\overline{\beta_l}}{z_{j_k}-\bar{z}_{j_l}} -\frac{\overline{\beta_l}}{z_{j_k}+\bar{z}_{j_l}}+\frac{iq_-\beta_l}{z_{j_l}z_{j_k}-1}-\frac{iq_-\beta_l}{z_{j_l}z_{j_k}+1}\right) ,
\end{align}
and
\begin{align}	
	 c_{i_s+N_1}^{-1}T(w_{i_s})^{-2}e^{-2i\theta(w_{i_s})t}\alpha_k&=\frac{i}{w_{i_s}}e^{i\nu_-(x,t;q^r_\Lambda)}q_-+\sum_{h=1}^{n_2}\left(\frac{\overline{\kappa_h}}{w_{i_s}-\bar{w}_{i_h}}+\frac{\overline{\kappa_h}}{w_{i_s}+\bar{w}_{i_h}} \right) \nonumber\\
	&+\sum_{l=1}^{n_1}\left( \frac{\overline{\varsigma_l}}{w_{i_s}-\bar{z}_{j_l}} +\frac{\overline{\varsigma_l}}{w_{i_s}+\bar{z}_{j_l}}-\frac{iq_-\varsigma_l}{z_{j_l}w_{i_s}-1}-\frac{iq_-\varsigma_l}{z_{j_l}w_{i_s}+1}\right),\\
	 c_{i_s+N_1}^{-1}T(w_{i_s})^{-2}e^{-2i\theta(w_{i_s})t}\kappa_k&=\frac{i}{w_{i_s}}e^{-i\nu_-(x,t;q^r_\Lambda)}\bar{q}_-+\sum_{h=1}^{n_2}\left(\frac{\overline{\alpha_h}}{w_{i_s}-\bar{w}_{i_h}}-\frac{\overline{\alpha_h}}{w_{i_s}+\bar{w}_{i_h}} \right) \nonumber\\
	&+\sum_{l=1}^{n_1}\left( \frac{\overline{\beta_l}}{w_{i_s}-\bar{z}_{j_l}} -\frac{\overline{\beta_l}}{w_{i_s}+\bar{z}_{j_l}}+\frac{iq_-\beta_l}{z_{j_l}w_{i_s}-1}-\frac{iq_-\beta_l}{z_{j_l}w_{i_s}+1}\right),
\end{align}
 for $k=1,...,n_1$, $s=1,...,n_2$ respectively.
\end{Proposition}

\begin{proof}
	The uniqueness of solution follows from the Liouville's theorem. Case I can be simple obtain. As for Case II, the symmetries  of $M^{(r)}_\Lambda(z)$  means that
it  admits a partial fraction expansion of following form as above. And to obtain $\beta_k$, $\varsigma_k$,$\alpha_s$ and $\kappa_s$, we substitute  (\ref{msol2}) into (\ref{resMrsol}) and obtain four linearly dependant equations set above.
\end{proof}
\begin{corollary}\label{sol}
When $\rho(s)\equiv0$, the scattering matrices $S(z)\equiv I$, which means $q_-=q_+$. Denote $q^r_\Lambda(x,t)$ is the $\mathcal{N}(\Lambda)$-soliton with   scattering data $\left\lbrace  0,\left\lbrace \zeta_n,\mathring{c}_n\right\rbrace_{n\in\Lambda}\right\rbrace$. By the reconstruction formula (\ref{recons u}) and (\ref{q}), the solution $q^r_\Lambda(x,t)$  of (\ref{DNLS}) with  scattering data $\left\lbrace  0,\left\lbrace \zeta_n,\mathring{c}_n\right\rbrace_{n\in\Lambda}\right\rbrace$   is given by:
\begin{align}
	q^r_\Lambda(x,t)=e^{i\nu_-(x,t;q^r_\Lambda)}\lim_{z\to \infty}z\left[M^{(r)}_\Lambda \right]_{12} .\label{qr}
\end{align}
Then in case I,
\begin{equation}
	u^r_\Lambda(x,t)=\lim_{z\to \infty}z|\left[M ^r_\Lambda\right]_{12}|= 1.
\end{equation}
So $\nu_-(x,t;q^r_\Lambda)=0$ and
\begin{align}
	q^r_\Lambda(x,t)=q_-.\label{q1}
\end{align}
And in case II,
\begin{align}
	u^r_\Lambda(x,t)&=\lim_{z\to \infty}z|\left[M^{(r)}_\Lambda \right]_{12}|\nonumber\\
	&=|ie^{i\nu_-(x,t;q^r_\Lambda)}q_-+2\sum_{s=1}^{n_2}\bar{\kappa}_k+2\sum_{k=1}^{n_1}(\bar{\varsigma}_k-iq_-\varsigma_k)|,
\end{align}
which leads to $\nu_-(x,t;q^r_\Lambda)=\frac{1}{2}\int_{-\infty}^x (|u^r_\Lambda(y,t)|^2-1)dy$ and
\begin{align}
	q^r_\Lambda(x,t)&=\lim_{z\to \infty}e^{i\nu_-(x,t;q^r_\Lambda)}z\left[M^{(r)}_\Lambda \right]_{12}\nonumber\\
	&=e^{2i\nu_-(x,t;q^r_\Lambda)}\left(ie^{i\nu_-(x,t;q^r_\Lambda)}q_-+2\sum_{s=1}^{n_2}\bar{\kappa}_k+2\sum_{k=1}^{n_1}(\bar{\varsigma}_k-iq_-\varsigma_k) \right) .\label{qr2}
\end{align}
\end{corollary}

\section{The small norm RH problem  for error function }\label{sec7}

\quad In this section,  we consider the error matrix-function $E(z)$ and  show that for large times, the error function $E(z)$ solves a small norm RH problem which  can be expanded asymptotically.
From the definition (\ref{transMr}), we can obtain a RH problem  for the matrix function  $E(z)$.

\noindent\textbf{RHP6}   Find a matrix-valued function $E(z)$  with following properties:

$\blacktriangleright$ Analyticity: $E(z)$ is analytical  in $\mathbb{C}\setminus  \Sigma^{(2)} $;

$\blacktriangleright$ Asymptotic behaviors:
\begin{align}
&E(z) \sim I+\mathcal{O}(z^{-1}),\hspace{0.5cm}|z| \rightarrow \infty;
\end{align}

$\blacktriangleright$ Jump condition: $E$ has continuous boundary values $E_\pm$ on $\Sigma^{(2)}$ satisfying
$$E_+(z)=E_-(z)V^{E},$$
 where the jump matrix $V^{E}$ is given by
\begin{equation}
V^{E}(z)=M^{(r)}_\Lambda(z)V^{(2)}(z)M^{(r)}_\Lambda(z)^{-1}. \label{VE}
\end{equation}

{Proposition \ref{unim}} implies that $M^{(r)}_\Lambda(z)$ is bound on $\Sigma^{(2)}$. By using Lemma \ref{lemmav2} and Corollary \ref{v2p}, we have the following estimates
\begin{equation}
\parallel V^{E}-I \parallel_{L^p} \lesssim \parallel V^{(2)}-I \parallel_{L^p}=\mathcal{O}(e^{- 2\rho_0t} ) , \label{VE-I}
\end{equation}
for $1\leq p \leq +\infty$.
This uniformly vanishing bound $\parallel V^{E}-I \parallel$ establishes RHP6 as a small-norm RH problem.
Therefore,    the   existence and uniqueness  of  the RHP6 can  shown  by using  a  small-norm RH problem
\begin{equation}
E(z)=I+\frac{1}{2\pi i}\int_{\Sigma^{(2)}}\dfrac{\left( I+\eta(s)\right) (V^E-I)}{s-z}ds,\label{Ez}
\end{equation}
where the $\eta\in L^2(\Sigma^{(2)})$ is the unique solution of following equation
\begin{equation}
(1-C_E)\eta=C_E\left(I \right),
\end{equation}
here $C_E$:$L^2(\Sigma^{(2)})\to L^2(\Sigma^{(2)})$ is a integral operator defined by
\begin{equation}
C_E(f)(z)=C_-\left( f(V^E-I)\right).
\end{equation}
The Cauchy projection operator $C_-$    on $\Sigma^{(2)}$  is
\begin{equation}
C_-(f)(s)=\lim_{z\to \Sigma^{(2)}_-}\frac{1}{2\pi i}\int_{\Sigma^{(2)}}\dfrac{f(s)}{s-z}ds.
\end{equation}
Then by (\ref{VE}) we have
\begin{equation}
\parallel C_E\parallel\leq\parallel C_-\parallel \parallel V^E-I\parallel_{L^\infty} \lesssim \mathcal{O}(e^{- 2\rho_0t} ),
\end{equation}
which means $\parallel C_E\parallel<1$ for sufficiently large t,   therefore  $1-C_E$ is invertible,  and   $\eta$  exists and is unique.
Moreover,
\begin{equation}
\parallel \eta\parallel_{L^2(\Sigma^{(2)})}\lesssim\dfrac{\parallel C_E\parallel}{1-\parallel C_E\parallel}\lesssim\mathcal{O}(e^{- 2\rho_0t} ).\label{normeta}
\end{equation}
Then we have the existence and boundedness of $E(z)$. In order to reconstruct the solution $q(x,t)$ of (\ref{DNLS}), we need the asymptotic behavior of $E(z)$ as $z\to \infty$.
\begin{Proposition}\label{asyE}
	For $E(z)$ defined in (\ref{Ez}), it stratifies
	\begin{equation}
		|E(z)-I|\lesssim\mathcal{O}(e^{- 2\rho_0t}) .
	\end{equation}
	As $z\to \infty$, the large $z$ expansion of $E$ is
	\begin{align}
	E(z)=I+E_1z^{-1}+\mathcal{O}(z^{-2}),\label{expE}
	\end{align}
	where
	\begin{equation}
	E_1=-\frac{1}{2\pi i}\int_{\Sigma^{(2)}}\left( I+\eta(s)\right) (V^{E}-I)ds,
	\end{equation}
	satisfying long time asymptotic behavior condition
	\begin{equation}
	E_1\lesssim\mathcal{O}(e^{- 2\rho_0t}).\label{E1t}
	\end{equation}
\end{Proposition}
\begin{proof}
	By combining (\ref{normeta}) and (\ref{VE-I}), we obtain
	\begin{equation}
		|E(z)-I|\leq|(1-C_E)(\eta)|+|C_E(\eta)|\lesssim\mathcal{O}(e^{- 2\rho_0t}).
	\end{equation}
	As $z\to \infty$,geometrically expanding $(s-z)^{-1}$
	for $z$ large in (\ref{Ez}) leads to (\ref{expE}). Finally for $E_1$,
	\begin{align}
		|E_1|\lesssim \parallel V^{E}-I \parallel_{L^1}+\parallel \eta \parallel_{L^2}2\parallel V^{E}-I \parallel_{L^2}\lesssim\mathcal{O}(e^{- 2\rho_0t}).
	\end{align}
\end{proof}

\section{Analysis  on  the pure $\bar{\partial}$-Problem}\label{sec8}
\quad Now we consider   the  asymptotics behavior of $M^{(3)}(z)$.
The $\bar{\partial}$-problem 4  of $M^{(3)}(z)$ is equivalent to the integral equation
\begin{equation}
M^{(3)}(z)=I+\frac{1}{\pi}\int_\mathbb{C}\dfrac{M^{(3)}(s)W^{(3)}(s)}{z-s}dm(s), \label{81}
\end{equation}
where $m(s)$ is the Lebesgue measure on the $\mathbb{C}$. Denote $C_z$ as the left Cauchy-Green integral  operator defined by
\begin{equation*}
fC_z(z)=\frac{1}{\pi}\int_C\dfrac{f(s)W^{(3)}(s)}{z-s}dm(s).
\end{equation*}
Then above equation (\ref{81}) can be rewritten as
\begin{equation}
M^{(3)}(z)=I\cdot\left(I-C_z \right) ^{-1}.\label{deM3}
\end{equation}
The existence of operator $\left(I-C_z \right) ^{-1}$ is given by the  following Lemma.
\begin{lemma}\label{Cz}
	The norm of the integral operator $C_z$ decay to zero as $t\to\infty$:
	\begin{equation}
	\parallel C_z\parallel_{L^\infty\to L^\infty}\lesssim  t^{-1/2},
	\end{equation}
	which implies that  $\left(I-C_z \right) ^{-1}$ exists.
\end{lemma}
\begin{proof}
	For any $f\in L^\infty$,
	\begin{align}
	\parallel fC_z \parallel_{L^\infty}&\leq \parallel f \parallel_{L^\infty}\frac{1}{\pi}\int_C\dfrac{|W^{(3)}(s)|}{|z-s|}dm(s),\nonumber
	\end{align}
where  $W (s)=M^{(r)}(z)\bar{\partial}R^{(2)}(z)M^{(r)}(z)^{-1}$.  So we only need to  estimate the integral
	\begin{equation*}
	\frac{1}{\pi}\int_C\dfrac{|W (s)|}{|z-s|}dm(s).
	\end{equation*}
	Since $W (s)\equiv0$ out of $\bar{\Omega}$,   we only need to focus on the estimate
$$\frac{1}{\pi}\int_\Omega\dfrac{|W (s)|}{|z-s|}dm(s). $$
Unlike the zero boundary case in \cite{fNLS},  here  $\det M^{(r)}(z)=1+z^{-2}$,  and Proposition \ref{asyE} implies that $| M^{(r)}(z)|\lesssim \sqrt{1+|z|^{-2}}$. So
	\begin{equation}
		\frac{1}{\pi}\int_\Omega\dfrac{|W (s)|}{|z-s|}dm(s)\lesssim \frac{1}{\pi}\int_\Omega\dfrac{|\bar{\partial}R^{(2)} (s)|}{|z-s|}\frac{1+|s|^{-2}}{|1+s^{-2}|}dm(s).
	\end{equation}
	Then for $j=1,4,5,8$, $| M^{(r)}(z)|$ is bounded in $\Omega_j$. But when  $z\in\Omega_j$ for $j=2,3,6,7$, the singularity at $z=\pm i$ need to be treat more carefully. So in following calculation, we take $\Omega_2$ in the second case as an example, because it is more elaborate than $\Omega_j$ for $j=1,4,5,8$. Denote three sub-region of $\Omega_2$ as
	\begin{align}
		D_1=\mathbb{D}(0,1-\epsilon_0)\cap\Omega_2&,\hspace{0.3cm}D_2=\mathbb{D}(0,1+\epsilon_0)\setminus\mathbb{D}(0,1-\epsilon_0)\cap \Omega_2,\nonumber\\
		&D_3=\Omega_2\setminus\mathbb{D}(0,1+\epsilon_0).
	\end{align}
Then the integral $\int_{\Omega_2}\dfrac{|W (s)|}{|z-s|}dm(s)$	is divide to three part:
\begin{align}
	I_i=\int_{D_i}\dfrac{|\bar{\partial}R^{(2)} (s)|}{|z-s|}\frac{1+|s|^{-2}}{|1+s^{-2}|}dm(s),\text{ for }i=1,2,3.	
\end{align}
Let $s=u+vi=re^{i\vartheta}$, $z=x+yi$. In the following calculation,  we will use the inequality
\begin{equation}
	\parallel |s-z|^{-1}\parallel_{L^q(\mathbb{R}^+)}=\left\lbrace \int^{+\infty}_{0}\left[  \left( \frac{v-y}{u-x}\right) ^2+1\right]  ^{-\frac{q}{2}} d\left( \frac{v-y}{|u-x|}\right)\right\rbrace ^{\frac{1}{q}}|u-x|^{-\frac{1}{p}} \lesssim |u-x|^{-\frac{1}{p}},
\end{equation}
with $1\leq q<+\infty$ and $\frac{1}{p}+\frac{1}{q}=1$.
For $s\in D_3$, $|s|>1+\epsilon_0$, then
	\begin{equation}
		\frac{1+|s|^{-2}}{|1+s^{-2}|}<\frac{1+|s|^2}{|s|^2-1}<1+\frac{2}{\epsilon_0^2+2\epsilon_0}<\infty.
	\end{equation}	
Then together with (\ref{R(2)}),  we have
	\begin{align}
	I_3\lesssim\int_{\Omega_2}\dfrac{|\bar{\partial}R^{(2)} (s)|}{|z-s|}dm(s)=\int_{\Omega_2}\dfrac{|\bar{\partial}R_2 (s)e^{2it\theta}|}{|z-s|}dm(s) \label{I_3}.
	\end{align}
Moreover, by \textbf{Lemma} \ref{Imtheta},
\begin{equation}
	|e^{2it\theta}|\leq e^{-c\sin 2\vartheta F(r)^2}\leq e^{-2cuv}\leq e^{-2cu},
\end{equation}
where $c$ is a positive constant, and the last step follows form
$$v\geq \max\left\lbrace 1+\epsilon_0,\frac{u}{\tan \varphi}\right\rbrace \geq 1+\epsilon_0>1.$$
Substitute (\ref{dbarRj}) and   above inequality into (\ref{I_3}) and obtain:
	\begin{align}
&I_3\lesssim\int_{0}^{+\infty}\int_{\frac{u}{\tan \varphi}}^{+\infty}\dfrac{|p_2' (ir)|e^{-4cut}}{|z-s|}dvdu+\int_{0}^{+\infty}\int_{\frac{u}{\tan \varphi}}^{+\infty}\dfrac{|r|^{-1/2}e^{-2cut}}{|z-s|}dvdu\nonumber\\
&+\int_{0}^{+\infty}\int_{\frac{u}{\tan \varphi}}^{+\infty}\dfrac{|\bar{\partial}X_1(r)|e^{-4cut}}{|z-s|}dvdu.\nonumber
	\end{align}
By Cauchy-Schwarz inequality, the first  item have
\begin{align}
	\int_{0}^{+\infty}\int_{\frac{u}{\tan \varphi}}^{+\infty}\dfrac{|p_2' (ir)|e^{-4cut}}{|z-s|}dvdu&\leq \int_{0}^{+\infty}\parallel\tilde{\rho}'\parallel_{L^2(i\mathbb{R})}\parallel |s-z|^{-1}\parallel_{L^2(\mathbb{R}^+)}e^{-2cut}du\nonumber\\
	\leq \int_{0}^{+\infty}e^{-2cut}|u-x|^{-\frac{1}{2}}du\lesssim t^{-\frac{1}{2}}.
\end{align}
So does the last item. Before we estimating the second item, we  consider for $p>2$,
\begin{align}
	\left( \int_{\frac{u}{\tan \varphi}}^{+\infty}|\sqrt{u^2+V^{(2)}}|^{-\frac{p}{2}}dv\right) ^{\frac{1}{p}}&=\left( \int_{\frac{u}{\sin \varphi}}^{+\infty}|r|^{-\frac{p}{2}+1}v^{-1}dr\right) ^{\frac{1}{p}}\lesssim u^{-\frac{1}{2}+\frac{1}{p}}.
\end{align}
Then
\begin{align}
	\int_{0}^{+\infty}\int_{\frac{u}{\tan \varphi}}^{+\infty}\dfrac{|r|^{-1/2}e^{-2cut}}{|z-s|}dvdu&\leq \int_{0}^{+\infty}\parallel|r|\parallel_{L_v^p(\frac{u}{\tan \varphi},+\infty)}\parallel |s-z|^{-1}\parallel_{L^q(\mathbb{R}^+)}e^{-2cut}du\nonumber\\
	&\leq \int_{0}^{+\infty}e^{-2cut}|u-x|^{-\frac{1}{p}}u^{-\frac{1}{2}+\frac{1}{p}}du\lesssim t^{-\frac{1}{2}}.
\end{align} 	
Combing above inequality we final have $I_3\lesssim t^{-\frac{1}{2}}$. As for $I_2$, the singularity  at $i$ can be balanced  by (\ref{Ri}), and recall that $1>\epsilon_0>0$ with $(1-\epsilon_0)\cos\varphi>\frac{1}{2}$
\begin{align}
	I_2&\lesssim\int_{0}^2\int_{1/2}^{2}\dfrac{e^{-2cut}}{|z-s|}\frac{1+|s|^2}{|s+i|}dvdu\lesssim \int_{0}^2\int_{1/2}^{2}\dfrac{e^{-2cut}}{|z-s|}dvdu\nonumber\\
	&\lesssim \int_{0}^2 |u-x|^{-1/2}e^{-2cut}du\lesssim |t|^{-1/2}.
\end{align}
Finally, consider $I_1$, similarly we have
\begin{align}
	I_1\lesssim \int_{0}^{1-\epsilon_0}\int_{u}^{1-\epsilon_0}\left( |p_2' (ir)|+|r|^{-1/2}+|\bar{\partial}X_1(r)|\right) \dfrac{e^{-2cut}}{|z-s|}dvdu,
\end{align}
which can be estimated same as $I_3$.  So the proof is completed.
\end{proof}
As $z\to\infty$, $M^{(3)}(z)$ has asymptotic expansion:
\begin{equation}
	M^{(3)}(z)=I-M^{(3)}_1(x,t)z+\mathcal{O}(z^{-2}),
\end{equation}
where $M^{(3)}_1$ is a $z$-independent coefficient.
The asymptotic behavior of $M^{(3)}_1$   given by following Proposition.

\begin{Proposition}\label{asyM3}
	As $z\to\infty$, the expansion above holds with
\begin{equation}
	M^{(3)}_1(x,t)=\frac{1}{\pi}\int_CM^{(3)}(s)W^{(3)}(s)dm(s).
\end{equation}
	There exist constants $T_1$, such that for all $t>T_1$, $M^{(3)}_1(x,t)$  satisfies
	\begin{equation}
		|M^{(3)}_1(x,t)|\lesssim t^{-3/4}.\label{M31}
	\end{equation}
\end{Proposition}
\begin{proof}
	{Lemma \ref{Cz}} and (\ref{deM3}) implies that for large $t$,   $\parallel M^{(3)} (z) \parallel_{L^\infty} \lesssim1$. The proof proceeds along
the same lines as the proof of above Proposition. For same reason, we only estimate the integral on $\Omega_2$.
	Like in the above Proposition,
	\begin{equation}
		\frac{1}{\pi}\int_{\Omega_2}M^{(3)}(s)W^{(3)}(s)dm(s)\lesssim\frac{1}{\pi}\int_{\Omega_2}|\bar{\partial}R_2(s)e^{2it\theta}|\frac{1+|s|^{-2}}{|1+s^{-2}|}dm(s).
	\end{equation}
Let $s=u+vi=re^{i\vartheta}$.	And	we also divide right integral of above  inequality to three part
	\begin{equation}
	I_{i+3}=\frac{1}{\pi}\int_{D_i}|\bar{\partial}R_2(s)e^{2it\theta}|\frac{1+|s|^{-2}}{|1+s^{-2}|}dm(s).
	\end{equation}	
	For $I_4$, $\frac{1+|s|^{-2}}{|1+s^{-2}|}<\infty$, so
	\begin{align}
		I_4\lesssim&\int_{0}^{+\infty}\int_{\frac{u}{\tan \varphi}}^{+\infty}|p_2'(ir)|e^{-2cuvt}dvdu+\int_{0}^{+\infty}\int_{\frac{u}{\tan \varphi}}^{+\infty}|r|^{-\frac{1}{2}}e^{-2cuvt}dvdu\nonumber\\
		&+\int_{0}^{+\infty}\int_{\frac{u}{\tan \varphi}}^{+\infty}|\bar{\partial}X_1(r)|e^{-2cuvt}dvdu.\label{I4}
	\end{align}
Note that
\begin{align}
	\left( \int_{\frac{u}{\tan \varphi}}^{+\infty}e^{-2cuvtq}dv\right) ^{\frac{1}{q}}&=	\left( \int_{\frac{u}{\tan \varphi}}^{+\infty}e^{-2cuvtq}d(2cuvtq)\right) ^{\frac{1}{q}}(2cutq)^{-\frac{1}{q}}\nonumber\\
	&\lesssim e^{-c'u^2t}(ut)^{-\frac{1}{q}},
\end{align}
where $c'$ is a positive constant.
Then the first integral in (\ref{I4}) have
\begin{align}
&\int_{0}^{+\infty}\int_{\frac{u}{\tan \varphi}}^{+\infty}|p_2'(ir)|e^{-2cuvt}dvdu  \nonumber\\
&\lesssim 	t^{-\frac{1}{2}}\int_{0}^{+\infty}\parallel\tilde{\rho}'\parallel_{L^2(i\mathbb{R})} u^{-\frac{1}{q}}e^{-c'u^2t}du\lesssim t^{-\frac{3}{4}}.\nonumber
\end{align}
The last integral can be bounded in same way. To estimate the second item, we also  use Cauchy-Schwarz inequality for $4>p>2$ and $\frac{1}{q}+\frac{1}{p}=1$
\begin{align}
	\int_{0}^{+\infty}\int_{\frac{u}{\tan \varphi}}^{+\infty}|r|^{-\frac{1}{2}}e^{-2cuvt}dvdu\lesssim t^{-\frac{1}{q}}\int_{0}^{+\infty}u^{\frac{2}{p}-\frac{3}{2}}e^{-c'u^2t}du\lesssim t^{-\frac{3}{4}}.
\end{align}
The bound for $I_4$ follows in the same manner as for $I_6$. Turning to $I_5$, we also use $|\bar{\partial}R_2(z)|\lesssim |z- i|$ and obtain
\begin{align}
	I_5&\lesssim\int_{D_2}\dfrac{e^{-2cut}}{|z-s|}\frac{1+|s|^2}{|s+i|}dm(s)\lesssim \int_{1/2}^2\int_{u}^{2}e^{-2cuvt}dvdu\nonumber\\
	&=\int_{1/2}^2(cut)^{-1}\left(e^{-2cu^2t}-e^{-4cut} \right)du \lesssim t^{-1}.
\end{align}
This estimate  is strong enough to obtain the result.
\end{proof}

\section{Asymptotic  for the  DNLS equation }\label{sec9}

\quad Now we begin to construct the long time asymptotics of the DNLS equation (\ref{DNLS}).
 Inverting the sequence of transformations (\ref{transm1}), (\ref{transm2}), (\ref{transm3}) and (\ref{transMr}), we have
\begin{align}
M(z)=&T(\infty)^{\sigma_3}M^{(3)}(z)E(z)M^{(r)}_\Lambda(z)R^{(2)}(z)^{-1}T(z)^{-\sigma_3}.
\end{align}
To  reconstruct the solution $q(x,t)$ by using (\ref{recons q}),   we take $z\to \infty$ out of $\bar{\Omega}$. In this case,  $ R^{(2)}(z)=I$. Further using   {Propositions} \ref{proT}, \ref{asyE}  and  \ref{asyM3},  we can obtain that
\begin{align}
	M(z)=&T(\infty)^{\sigma_3}\left(I+ M^{(3)}_1(z)z^{-1}\right) E(z)M^{(r)}_\Lambda(z)\nonumber\\
	&T(\infty)^{-\sigma_3}\left(1+z^{-1}\frac{1}{2\pi i}\int _{\mathbb{R}}\log (1-\rho(s)\tilde{\rho}(s))ds\right)^{-\sigma_3} + \mathcal{O}(z^{-2}),
\end{align}
whose admits  long time asymptotics
\begin{align}
&	M(z)=&T(\infty)^{\sigma_3} M^{(r)}_\Lambda(z)T(\infty)^{-\sigma_3}\left(1+z^{-1}\frac{1}{2\pi i}\int _{\mathbb{R}}\log (1-\rho(s)\tilde{\rho}(s))ds\right)^{-\sigma_3}\nonumber\\
& + \mathcal{O}(z^{-2})+\mathcal{O}(t^{-3/4}).\nonumber
\end{align}
From (\ref{recons u}),
\begin{align}
|m(x,t)|&=| \lim_{z\to \infty} z\left[   M (z) \right] _{12}|=|T(\infty)^{-2}||\lim_{z\to \infty} z\left[   M^{(r)}_\Lambda(z) \right] _{12}\nonumber\\
&+\lim_{z\to \infty}\left[   M^{(r)}_\Lambda(z) \right]_{12}\frac{1}{2\pi i}\int _{\mathbb{R}}\log (1-\rho(s)\tilde{\rho}(s))ds|+\mathcal{O}(t^{-3/4})\nonumber\\
&=|q^r_\Lambda(x,t)|+\mathcal{O}(t^{-3/4}),
\end{align}
where  $q^r_\Lambda(x,t)$  is given  in Corollary \ref{sol}.  Then from (\ref{q}),
\begin{align}
	q(x,t)&= exp\left\lbrace \frac{i}{2}\int_{-\infty}^x( |m(x,t)|^2-1)dy \right\rbrace m(x,t)\nonumber\\
	&=exp\left\lbrace \frac{i}{2}\int_{-\infty}^x( |q^r_\Lambda(x,t)|^2-1)dy \right\rbrace T(\infty)^{-2}q^r_\Lambda(x,t)+\mathcal{O}(t^{-3/4}).
\end{align}
Therefore, we achieve main result of this paper.

\begin{theorem}\label{last}   Let $q(x,t)$ be the solution for  the initial-value problem (\ref{DNLS}) with generic data   $u_0(x)\in H^{1,1}(\mathbb{R})$ and scatting data $\left\lbrace  r(z),\left\lbrace \zeta_n,C_n\right\rbrace^{4N_1+2N_2}_{n=1}\right\rbrace$. Let $\xi=\frac{x}{t}$ with $-3<\xi<-1$.
Denote $q^r_\Lambda(x,t)$ be the $\mathcal{N}(\Lambda)$-soliton solution corresponding to   scattering data
$\left\lbrace  0,\left\lbrace \zeta_n,\mathring{c}_n\right\rbrace_{n\in\Lambda}\right\rbrace$ shown in Corollary \ref{sol}. And $\Lambda$ is  defined in (\ref{devide}). There exist a large constant $T_1=T_1(\xi)$, for all $T_1<t\to\infty$,
	\begin{align}
	q(x,t)=exp\left\lbrace \frac{i}{2}\int_{-\infty}^x(|q^r_\Lambda(x,t)|^2-1)dy \right\rbrace T(\infty)^{-2}q^r_\Lambda(x,t)+\mathcal{O}(t^{-3/4}),\label{resultu}
	\end{align}
where $u^r_\Lambda(x,t)$ and $T(z)$ are show in Propositions \ref{proT} and Corollary \ref{sol} respectively.

\end{theorem}
The   long time asymptotic expansion  (\ref{resultu}) shows the soliton resolution  of  for  the initial value problem  of  the the derivative nonlinear schr$\ddot{o}$dinger equation,
which   can be characterized with  an $\mathcal{N}(\Lambda)$-soliton whose parameters are modulated by
a sum of localized soliton-soliton
 interactions . Our results also show that the poles on curve soliton solutions of
  short-pulse  equation has dominant contribution to the solution as $t\to\infty$.\vspace{6mm}

\noindent\textbf{Acknowledgements}

This work is supported by  the National Science
Foundation of China (Grant No. 11671095,  51879045).

\hspace*{\parindent}
\\


\begin{thebibliography}{10}


\bibitem{Manakov1974} S.V. Manakov, Nonlinear Fraunhofer diffraction,  Sov. Phys.-JETP 38(1974), 693-696.


\bibitem{ZM1976}
V. E. Zakharov, S. V.  Manakov,
\newblock {Asymptotic behavior of nonlinear wave systems integrated by the inverse scattering method},
\newblock {\em  Soviet Physics JETP,}   44(1976), 106-112.



\bibitem{SPC}
P. C. Schuur,
\newblock {Asymptotic analysis of soliton products},
\newblock {\em Lecture Notes in Mathematics,}  1232, 1986.


\bibitem{BRF}
R. F. Bikbaev,
\newblock {Asymptotic-behavior as t-infinity of the solution to the cauchy-problem for the landau-lifshitz equation},
\newblock {\em  Theor. Math. Phys,}  77(1988), 1117-1123.

\bibitem{Foka}
R. F. Bikbaev,
\newblock {Soliton generation for initial-boundary-value problems,}
\newblock {\em  Phys. Rev. Lett.}, 68(1992), 3117-3120.





\bibitem{RN6}
X. Zhou, P. Deift,
\newblock  A steepest descent method for oscillatory Riemann-Hilbert problems.
\newblock {\em Ann. Math.}, 137(1993),  295-368.


\bibitem{RN9}
X. Zhou, P. Deift,
\newblock  Long-time behavior of the non-focusing nonlinear Schr$\ddot{o}$dinger equation--a case study,
\newblock {\em Lectures in Mathematical Sciences}, Graduate School of Mathematical Sciences, University of Tokyo, 1994.


\bibitem{RN10}
P. Deift, X. Zhou,
\newblock Long-time asymptotics for solutions of the NLS equation with initial data in a weighted Sobolev space,
\newblock {\em Comm. Pure Appl. Math.}, 56(2003), 1029-1077.




\bibitem{Grunert2009}
K. Grunert,  G. Teschl,
\newblock   Long-time asymptotics for  the Korteweg de Vries equation  via  noninear  steepest descent.
\newblock {\em Math. Phys. Anal. Geom.},     12(2009), 287-324.

\bibitem{MonvelCH}
A. B. de Monvel, A. Kostenko, D. Shepelsky, G. Teschl,
\newblock  Long-time asymptotics for the Camassa-Holm equation.,
\newblock {\em SIAM J. Math. Anal}, 41(2009),  1559-1588.


\bibitem{xu2015}
J. Xu, E. G. Fan,
\newblock  Long-time asymptotics for the Fokas-Lenells equation with decaying initial value problem: Without solitons,
\newblock {\em  J. Differential Equations,}    259(2015), 1098-1148.


\bibitem{xusp}
J. Xu,
\newblock  Long-time asymptotics for the short pulse equation,
\newblock {\em  J. Differential Equations,}  265(2018), 3494-3532.



\bibitem{MandM2006}
K. T. R. McLaughlin, P. D. Miller,
\newblock {The $\bar{\partial}$ steepest descent method and the asymptotic behavior of polynomials orthogonal on
the unit circle with fixed and exponentially varying non-analytic weights},
\newblock {\em Int.  Math. Res. Not.}, (2006), Art. ID 48673.

\bibitem{MandM2008}
K. T. R. McLaughlin, P. D. Miller,
\newblock {The $\bar{\partial}$ steepest descent method for orthogonal polynomials on the real line with varying weights},
\newblock {\em  Int. Math. Res. Not.},  (2008), Art. ID   075.

\bibitem{DandMNLS}
M. Dieng, K. D. T. McLaughlin,
\newblock {Dispersive asymptotics for linear and integrable equations by the Dbar steepest descent method},
\newblock {\em } Nonlinear dispersive partial differential equations and inverse scattering,
253-291, Fields Inst. Commun., 83, Springer, New York,  2019

\bibitem{fNLS}
M. Borghese, R. Jenkins, K. T. R. McLaughlin, Miller P,
\newblock { Long-time asymptotic behavior of the focusing nonlinear Schr$\ddot{o}$dinger equation, }
\newblock {\em  Ann. I. H. Poincar$\acute{e}$ Anal}, 35(2018), 887-920.



\bibitem{Liu3}
R. Jenkins, J. Liu, P. Perry, C. Sulem,
\newblock   Soliton resolution for the derivative nonlinear Schr$\ddot{o}$dinger equation,
\newblock {\em Commun. Math. Phys.},  363(2018), 1003-1049.

\bibitem{SandRNLS}
S. Cuccagna, R. Jenkins,
\newblock {On asymptotic stability of N-solitons of the defocusing nonlinear Schr$\ddot{o}$dinger equation, }
\newblock {\em  Comm. Math. Phys}, 343(2016), 921-969.


\bibitem{KN1978}
D. J. Kaup and A. C. Newell,
\newblock  An exact solution for a Derivative Nonlinear Schr$\ddot{o}$dinger equation,
\newblock {\em  J. Math. Phys., }  19(1978), 798-801.



\bibitem{R1971}
A. Rogister,
\newblock  Parallel propagation of nonlinear low-frequency waves in high-$\beta$ plasma,
\newblock {\em Phys. Fluids.}, 14(1971), 2733-2739.	

\bibitem{Mj1976}
E. Mj$\phi$lhus,
\newblock   On the modulational instability of hydromagnetic waves parallel to the magnetic field,
\newblock {\em J. Plasma Phys.}, 16(1976), 321-334.


\bibitem{Mi1976}
K. Mio, T. Ogino, K. Minami, S. Takeda,
\newblock   Modified nonlinear Schr$\ddot{o}$dinger equation for Alfv$\acute{e}$n waves propagating along the magnetic field in cold plasmas,
\newblock {\em J. Phys. Soc. Jpn.}, 41(1976), 265-271.

\bibitem{Mj1989}
E. Mj$\phi$lhus,
\newblock    Nonlinear Alfv$\acute{e}$n waves and the DNLS equation: oblique aspects,
\newblock {\em Phys. Scr.}, 40(1989), 227-237.

\bibitem{Mj1997}
E. Mj$\phi$lhus, T. Hada, in: T. Hada, H. Matsumoto (Eds.),
\newblock    Nonlinear Waves and Chaos in Space Plasmas,
\newblock {\em Terrapub, Tokio},  (1997).


\bibitem{GP}
G. P. Agrawal,
\newblock  Nonlinear Fiber Optics.
\newblock {\em Academic Press, Boston}, 1989.

\bibitem{DM1983}
D. Anderson and M. Lisak,
\newblock    Nonlinear Asymmetric Self-phase modulation and self-steepening of pulses in long Optical Waveguides,
\newblock {\em Phys. Rev. A.},  27(1983), 1393-1398.

\bibitem{NM1981}
N.Tzoar and M.Jain,
\newblock    Self-phase Modulation in long-geometry optical waveguide,
\newblock {\em Phys. Rev. A.},  23(1981), 1266-1270.

\bibitem{N1991}
I. Nakata,
\newblock  Weak nonlinear electromagnetic waves in a ferromagnet propagating parallel to an external magnetic field,
\newblock {\em  J. Phys. Soc. Jpn.},  60(1991), 3976-3977.

\bibitem{N1993}
I. Nakata, H. One and M. Yosida,
\newblock  Solitons in a dielectric medium under an external magnetic field,
\newblock {\em  Prog. Theor. Phys.},  90(1993), 739-742.

\bibitem{DV2002}
M. Daniel, V. Veerakumar,
\newblock  Propagation of electromagnetic soliton in antiferromagnetic medium,
\newblock {\em  Phys. Lett. A},  302(2002), 77-86.



\bibitem{ZH2007}
G. Q. Zhou, N. N. Huang
\newblock  An N-soliton solution to the DNLS equation based on revised inverse scattering transform,
\newblock {\em  J. Phys. A: Math. Theor., }  40(2007), 13607.

\bibitem{KI1978}
T. Kawata,  H. Inoue,
\newblock  Exact solutions of the derivative nonlinear Schr$\ddot{o}$dinger equation under the nonvanishing conditions,
\newblock {\em   J. Phys. Soc. Jpn., }  44(1978), 1968-1976.

\bibitem{CL2004}
X. J. Chen, W. K. Lam,
\newblock  Inverse scattering transform for the derivative nonlinear Schr$\ddot{o}$dinger equation with nonvanishing boundary conditions,
\newblock {\em    Phys. Rev. E, }  69(2004), 066604.

\bibitem{C2006}
X. J. Chen, J. Yang, W. K. Lam,
\newblock  N-soliton solution for the derivative nonlinear Schr$\ddot{o}$dinger equation with nonvanishing boundary conditions,
\newblock {\em  J. Phys. A, }  39(2006), 3263.

\bibitem{L2007}
 V. Lashkin,
\newblock    N-soliton solution  and perturbation theory for the derivative nonlinear Schr$\ddot{o}$dinger equation with nonvanishing boundary conditions,
\newblock {\em J. Phys. A},  40(2007), 6119.

\bibitem{ZGQ}
G. Zhang, Z. Yan,
\newblock   The Derivative nonlinear Schr$\ddot{o}$dinger equation with zero/Nonzero Bboundary conditions: inverse scattering
transforms and N-double-pole solutions,
\newblock {\em J. Non. Sci.},  2(2002).



   \bibitem{xf2013}  J. Xu and E. G. Fan, Inverse scattering for  the derivative nonlinear Schrodinger equation: A Riemann-Hilbert approach,  arXiv:1209.4245v1.

\bibitem{Xf2014}
J. Xu, E. Fan, Y. Chen,
\newblock   Long-time asymptotic for for the derivative nonlinear Schr$\ddot{o}$dinger equation  with step-like initial value,
\newblock {\em Math. Phys. Anal. Geometry},  16(2013), 253-288.








\end{thebibliography}
\end{document}